%% file: main.tex
\documentclass[11pt, a4paper]{article}

\usepackage{geometry}
\input{macro.tex}

\begin{document}
\begin{CJK}{UTF8}{gbsn}

\title{A Bertrand duopoly game with differentiated products reconsidered}

\author[a]{Xiaoliang Li}

\author[b]{Bo Li\thanks{Corresponding author: libomaths@163.com}}

\affil[a]{School of Digital Economics, Dongguan City University, Dongguan 523419, China}

\affil[b]{School of Finance, Anhui University of Finance and Economics, Bengbu 233030, China}

\date{}
\maketitle

\begin{abstract}
In this paper, we explore a dynamic Bertrand duopoly game with differentiated products, where firms are boundedly rational and consumers are assumed to possess an underlying CES utility function. We mainly focus on two distinct degrees of product substitutability. Several tools based on symbolic computations such as the triangular decomposition method and the PCAD method are employed in the analytical investigation of the model. The uniqueness of the non-vanishing equilibrium is proved and rigorous conditions for the local stability of this equilibrium are established for the first time. Most importantly, we find that increasing the substitutability degree or decreasing the product differentiation has an effect of destabilization for our Bertrand model, which is in contrast with the relative conclusions for the Cournot models. This finding could be conducive to the revelation of the essential difference between dynamic Cournot and Bertrand oligopolies with differentiated goods. In the special case of identical marginal costs, we derive that lower degrees of product differentiation mean lower prices, higher supplies, lower profits, and lower social welfare. Furthermore, complex dynamics such as periodic orbits and chaos are reported through our numerical simulations. 

\vspace{10pt}
\end{abstract}

\section{Introduction}

It is well known that Cournot \cite{Cournot1838R} developed the first formal theory of oligopoly, which is a market supplied by only a few firms. In Cournot's framework, firms are supposed to make decisions on their quantities of outputs and have perfect information on their rivals' strategic behavior. In the strand of Cournot oligopoly models, the market demand function is usually supposed to be linear for simplicity by many economists (e.g., Fisher \cite{Fisher1961T}, McManus and Quandt\cite{McManus1961C}). In the real world, however, a non-linear demand is more likely to exist. Puu \cite{Puu1991C} investigated a Cournot duopoly game under an isoelastic market demand, where the price is simply the reciprocal of the total supply. Afterward, fruitful contributions including \cite{Ahmed2000O, Askar2014T, Bischi2007O, Canovas2018O, Cavalli2015N, Elsadany2017D, Kopel1996S, Li2020N, Li2022C, Li2022A, Matsumoto2022N, Naimzada2009C}, were made in the literature on Cournot games. Related to our study, Zhang and Zhang \cite{Zhang1996S} considered a Cournot game in which each firm produces multiple products and sells them in multiple markets. They obtained sufficient and necessary conditions for the local stability of the Cournot-Nash equilibria.

Several decades later after Cournot's seminal work, Bertrand \cite{Bertrand1883R} proposed a different framework to describe oligopolistic competition, where prices rather than quantities are the strategic variables of the competitors. Singh and Vives \cite{Singh1984P} analyzed the duality of prices and quantities, and found that Cournot (Bertrand) competition with substitutes is the dual of Bertrand (Cournot) competition with complements. López and Naylor \cite{Lopez2004T} compared Cournot and Bertrand equilibria in a downstream differentiated duopoly, and proved that the classic conclusion that profits under Cournot equilibrium exceed those under Bertrand competition could be reversible in the case of imperfect substitutes. Zhang et al.\ \cite{Zhang2009T} considered a Bertrand model formulated under a linear inverse demand, and obtained the existence and stability of the equilibrium. Different from \cite{Zhang2009T}, Fanti et al.\ \cite{Fanti2013T} developed a model with sound microeconomic foundations that determine the demand for differentiated products, and showed that synchronized dynamics and intermittency phenomena may appear. Naimzada and Tramontana \cite{Naimzada2012D} also considered a Cournot-Bertrand duopoly model with product differentiation and emphasized the role of best response dynamics and an adaptive adjustment mechanism for stability. Brianzoni et al.\ \cite{Brianzoni2015D} assumed quadratic costs in the study of the Bertrand duopoly game with horizontal product differentiation and discovered synchronized dynamics. Moreover, Ma and Guo \cite{ma_i2016} studied the impacts of information on the dynamical Bertrand game. They showed that there exists a fixed point independent of the amount of information for a triopoly, and the stable region of adjustment parameter increases with the amount of information for a duopoly.

%Moreover, \cite{Zhou2021C} introduced the product service factor into the exploration of a Bertrand duopoly game with two firms producing differentiated commodities and found that keeping the adjustment speeds in a relatively small range may lead to a long-term stable operation of the market.

%\memo{differentiated goods}
%\cite{Economides1989S}
%\cite{Hoernig2003E}
%\cite{Symeonidis2003C}
%\cite{Fanti2012T}
%
%\cite{Askar2018Q}

In all the aforementioned Bertrand games, the inverse demand function is supposed to be linear. Instead, Gori and Sodini \cite{Gori2017P} explored the local and global dynamics of a Bertrand duopoly with a nonlinear demand and horizontal product differentiation. Furthermore, Ahmed et al.\ \cite{Ahmed2015O} proposed a dynamic Bertrand duopoly game with differentiated products, where firms are boundedly rational and consumers are assumed to possess an underlying CES utility function. They only employed numerical simulations to investigate the dynamic behavior of their model because the closed form of the equilibrium is extremely difficult to compute. They observed that the Nash equilibrium loses its stability through a period-doubling bifurcation as the speed of adjustment increases. Motivated by \cite{Ahmed2015O}, Agliari et al.\ \cite{Agliari2016N} investigated a Cournot duopoly game with differentiated goods. We should mention that Agliari et al.\ \cite{Agliari2016N} used the same CES utility function as \cite{Ahmed2015O} to derive the demand function of the market. {They discovered that a low degree of product substitutability or a higher degree of product differentiation may destabilize the Cournot game. This finding is in accordance with that of Fanti and Gori \cite{Fanti2012T}, where the authors introduced a Cournot duopoly with a linear demand and heterogeneous players to study the influence of product differentiation on stability and found that a higher degree of product differentiation may destabilize the market equilibrium.} 
% They pointed out that the possible reason is that a higher degree of product differentiation may imply weaker competition that tends to destabilize the equilibrium.

In this paper, we re-study the Bertrand duopoly game of Ahmed et al.\ \cite{Ahmed2015O} using several tools based on symbolic computations such as the triangular decomposition method (see, e.g., \cite{Li2010D}) and the PCAD method (see, e.g., \cite{Collins1991P}). It is worth noting that the results of symbolic computations are exact, and thus can provide theoretical foundations for the systematic analysis of economic models. We analytically investigate the local stability and bifurcations of the model. By using several tools based on symbolic computations, the uniqueness of the non-vanishing equilibrium is proved and the rigorous conditions for the local stability of this equilibrium are obtained for the first time. In the special case that the two companies have identical marginal costs, we prove that the model can lose its stability only through a period-doubling bifurcation. The most important finding is that increasing the substitutability degree or decreasing the product differentiation has an effect of destabilizing the unique non-vanishing equilibrium. A possible explanation is that a decrease in product differentiation may result in an increase in market competition intensity and even a price war, which could lead to the destabilization of the equilibrium. It should be noted that our finding is in contrast with the relative conclusions by Agliari et al.\ \cite{Agliari2016N} and by Fanti and Gori \cite{Fanti2012T}. This contradiction contributes to the literature on the connection between Cournot and Bertrand oligopolies and may help reveal the essential difference between them. In the special case of identical marginal costs, we derive the fact that lower degrees of product differentiation can lead to lower prices, higher supplies, lower profits, and lower social welfare. This fact is in line with our economic intuition. Complex dynamics such as periodic orbits and chaos can be observed through our numerical simulations, which also confirm that an increase in the substitutability degree leads to the emergence of instability in the considered model. Furthermore, we discover the existence of a Neimark-Sacker bifurcation directly on the equilibrium, which is a new finding and has not yet been discovered by Ahmed et al.\ \cite{Ahmed2015O}

The rest of this paper is structured as follows. In Section 2, we revisit the construction of the Bertrand duopoly game investigated in our study. We analytically explore the stability and bifurcations of this model for two different substitutability degrees, namely $\alpha=1/2$ and $\alpha=1/3$, in Sections 3 and 4, respectively. The influence of the substitutability degree on the local stability of the equilibrium and related comparative statics are discussed in Section 5. Numerical simulations are provided in Section 6. Concluding remarks are given in Section 7.

\section{Model}

In our study, we consider a market where two firms compete with each other and produce differentiated goods. The prices and quantities of the two goods are denoted by $p_i$ and $q_i$, respectively, with $i=1,2$. Furthermore, it is assumed that the market possesses a continuum of identical consumers with a CES utility function of the form
$$U(q_1,q_2)=q_1^\alpha + q_2^\alpha,$$
where $\alpha$ ($0<\alpha<1$) is called the substitutability degree between the products. Consumers choose their consumptions by maximizing the utility subject to the budget constraint
$$p_1q_1+p_2q_2=1.$$

Consequently, we have the following demand functions (The reader can refer to \cite{Ahmed2015O} for the proof).
$$q_1=\frac{p_2^\beta}{p_1} \frac{1}{p_1^\beta+p_2^\beta},
~~q_2=\frac{p_1^\beta}{p_2} \frac{1}{p_1^\beta+p_2^\beta},$$
where
$\beta=\alpha/(1-\alpha)$. Thus, the inverse demands of the two goods are
\begin{equation}\label{eq:inv-dem}
	p_1=\frac{q_1^{\alpha-1}}{q_1^{\alpha}+q_2^{\alpha}},
~~p_2=\frac{q_2^{\alpha-1}}{q_1^{\alpha}+q_2^{\alpha}}.
\end{equation}

Accordingly, a decrease in $\alpha$ would make the products less substitutable or more differentiated. In particular, if $\alpha=0$, the inverse demands become $p_1=\frac{1}{2\,q_1}$ and $p_2=\frac{1}{2\,q_2}$, which means that the two goods are completely independent. If $\alpha=1$, we obtain the inverse demand $p_1=p_2=\frac{1}{q_1+q_2}$, which is the same as the famous isoelastic demand function introduced by Puu \cite{Puu1991C}. In this case, the prices of the two goods are equal. That is to say, the two commodities are regarded as indistinguishable or identical by consumers.

The cost functions are assumed to be linear, i.e.,
$$C_1(q_1)=c_1q_1,~~~C_2(q_2)=c_2q_2,$$
where
$c_1>0$ and $c_2>0$. Then the profit of firm $i$ ($i=1,2$) should be
\begin{equation}\label{eq:profit}
	\Pi_i(p_i,p_{-i})=p_iq_i-c_iq_i=(p_i-c_i)\frac{p_{-i}^\beta}{p_i} \frac{1}{p_i^\beta+p_{-i}^\beta},
\end{equation}
where $p_{-i}$ denotes the price of the commodity produced by the rival.

Furthermore, the gradient adjustment mechanism is formulated as
$$p_i(t+1)=p_i(t)+k_i\frac{\partial \Pi_i(t)}{\partial p_i(t)},$$
where $k_i>0$ controls the adjustment speed of firm $i$.
It is known that 
$$\frac{\partial \Pi_i}{\partial p_i}=\frac{-p_{-i}^{\beta} p_i^{1+\beta} \beta+\left(p_{-i}^{2 \beta}+p_{-i}^{\beta} p_i^{\beta} \left(1+\beta\right)\right) c_i}{p_i^{2} \left(p_i^{\beta}+p_{-i}^{\beta}\right)^{2}}.$$

In short, the model can be described as the following iteration map.
\begin{equation}\label{eq:iter-map}
\left\{
\begin{split}
&p_1(t+1)=p_1(t)+k_1\frac{-p_2^{\beta}(t) p_1^{1+\beta}(t) \beta+\left(p_2^{2 \beta}(t)+p_2^{\beta}(t) p_1^{\beta}(t) \left(1+\beta\right)\right) c_1}{p_1^{2}(t) \left(p_1^{\beta}(t)+p_2^{\beta}(t)\right)^{2}},\\
&p_2(t+1)=p_2(t)+k_2\frac{-p_1^{\beta}(t) p_2^{1+\beta}(t) \beta+\left(p_1^{2 \beta}(t)+p_1^{\beta}(t) p_2^{\beta}(t) \left(1+\beta\right)\right) c_2}{p_2^{2}(t) \left(p_2^{\beta}(t)+p_1^{\beta}(t)\right)^{2}}.
\end{split}\right.
\end{equation}

This game was first explored by Ahmed et al.\ \cite{Ahmed2015O} only through numerical simulations because no analytical expressions of the Nash equilibria are available. In this paper, we reconsider this game using methods based on symbolic computations and explore the influence of the substitutability degree on the local stability of the equilibrium. One can see that for general $\beta$, it is impossible to analyze the equilibrium point of map \eqref{eq:iter-map}, because the system will have an exponential parameter. For such systems with exponential parameters, existing analytical tools are quite limited. Therefore, similar to \cite{Ahmed2015O}, our study mainly focuses on two specific cases, namely $\beta=1$ and $\beta=1/2$, which are corresponding to $\alpha=1/2$ and $\alpha=1/3$, respectively.

\section{$\alpha=1/2$}

If $\alpha=1/2$, then $\beta=1$.  Hence, map \eqref{eq:iter-map} becomes
\begin{equation}\label{eq:map-12}
\left\{
\begin{split}
&p_1(t+1)=p_1(t)+k_1\frac{-2\,p_2(t) p_1^{2}(t)+\left(p_2^{2}(t)+2\,p_2(t) p_1(t)\right) c_1}{p_1^{2}(t) \left(p_1(t)+p_2(t)\right)^{2}},\\
&p_2(t+1)=p_2(t)+k_2\frac{-2\,p_1(t) p_2^{2}(t)+\left(p_1^{2}(t)+2\, p_1(t) p_2(t)\right) c_2}{p_2^{2}(t) \left(p_2(t)+p_1(t)\right)^{2}}.
\end{split}\right.
\end{equation}

From an economic point of view, it is important to identify the number of non-vanishing equilibria $(p_1,p_2)$ with $p_1>0$ and $p_2>0$.
In order to compute the equilibrium, we set $p_1(t+1)=p_1(t)=p_1$ and $p_2(t+1)=p_2(t)=p_2$. Then the following equations of the equilibrium are acquired. 
\begin{equation}\label{eq:equi-a12}
\left\{
\begin{split}
	-2p_2 p_1^{2}+\left(p_2^{2}+2p_2 p_1\right) c_1=0,\\
	-2p_1 p_2^{2}+\left(p_1^{2}+2p_1 p_2\right) c_2=0.
\end{split}
\right.
\end{equation}

%\gai{From an economic point of view, it is important to identify the number of non-vanishing equilibria $(p_1,p_2)$ with $p_1>0$ and $p_2>0$. Based on the triangular decomposition method, the first author of this paper and his coworker proposed an algebraic approach to tackle the multiplicity of equilibria (see \cite{Li2014C}). It is noted that this approach can systematically handle economies with inequality conditions. In the sequel, we explain the main idea of our approach and use it to identify the number of non-vanishing equilibria of map \eqref{eq:map-12}.}

%By \eqref{eq:equi-a12}, one can find the zero equilibrium $(p_1,p_2)=(0,0)$. However, it is not a simple task to compute the non-vanishing equilibria. 

The triangular decomposition method, which can be viewed as an extension of the Gaussian elimination method, permits us to analyze the equilibria of non-linear economic models. Both the method of triangular decomposition and the method of Gaussian elimination can transform a system into triangular forms. However, the triangular decomposition method is feasible for polynomial systems, while the Gaussian elimination method is just for linear systems. Refer to \cite{Aubry1999T, Kalkbrener1993A, Li2010D, Wang2000C, Wu1986B} for more information on triangular decomposition. Specifically, using the triangular decomposition method, we can decompose the solutions of system \eqref{eq:equi-a12} into zeros of the following two triangular polynomial sets.
\begin{equation}\label{eq:ts12}
\begin{split}
	&\pset{T}_{11}=\left[p_1,~p_2\right],\\
	&\pset{T}_{12}=\left[p_1^{3}-4\,c_1p_1^{2}+(4\,c_1^2-2\,c_1c_2)p_1+3\,c_1^2c_2,~c_1p_2-p_1^{2}+2\,c_1p_1\right].		
\end{split}
\end{equation}

The zero of $\pset{T}_{11}$ is corresponding to the origin $(0,0)$. Moreover, the non-vanishing equilibria can be computed from $\pset{T}_{12}$. The first polynomial $p_1^{3}-4\,c_1p_1^{2}+(4\,c_1^2-2\,c_1c_2)p_1+3\,c_1^2c_2$ of $\pset{T}_{12}$ is univariate in $p_1$ and the second polynomial $c_1p_2-p_1^{2}+2\,c_1p_1$ of $\pset{T}_{12}$ has degree $1$ with respect to $p_2$. Consequently, if we solve $p_1$ from the first polynomial, then we can substitute the solution of $p_1$ into the second polynomial and easily obtain $p_2$. As the first polynomial of $\pset{T}_{12}$ has degree $3$ with respect to $p_1$, we know that there are at most $3$ positive real solutions. Their analytical expressions exist but are quite complicated, though.

This is not an easy task to identify the exact number of positive real solutions if the analytical solutions of $\pset{T}_{12}$ are complicated. However, the first author of this paper and his co-worker \cite{Li2014C} proposed an algebraic algorithm to systematically identify multiplicities of equilibria in semi-algebraic economies without obtaining the closed-form solutions. We summarize the computational results for map \eqref{eq:map-12} in Proposition \ref{prop:unique-sol-12}. Interested readers can refer to Section 3 of \cite{Li2014C} for additional details of the algorithm.

%However, the algebraic approach proposed by Li and Wang \cite{Li2014C} can be used to systematically count real solutions of a semi-algebraic system without obtaining the closed-form solutions. Roughly speaking, the main idea of this approach is to compute the so-called border polynomial of
%\begin{equation}\label{eq:sys-mean-equi}
%\left\{
%\begin{split}
%&p_1^{3}-4\,c_1p_1^{2}+(4\,c_1^2-2\,c_1c_2)p_1+3\,c_1^2c_2=0,\\
%&c_1p_2-p_1^{2}+2\,c_1p_1=0,\\
%&p_1>0,~p_2>0,		
%\end{split}
%\right.
%\end{equation}
%which is
%$$c_1c_2\left(32\,c_1^2+61\,c_1c_2+c_2^2\right).$$
%This border polynomial divides the parameter set of our concern, i.e., $\{(c_1,c_2)\,|\,c_1>0,c_2>0\}$, into separated regions such that in each region the number of distinct real solutions of \eqref{eq:sys-mean-equi} is invariant. Therefore, the problem of counting solutions could be done region by region. 

\begin{proposition}\label{prop:unique-sol-12}
	Let $\alpha=1/2$. The iteration map \eqref{eq:map-12} possesses one unique equilibrium $(p_1,p_2)$ with $p_1>0$ and $p_2>0$.
\end{proposition}

To explore the local stability of the equilibrium, the following Jacobian matrix plays an ambitious role.
$$
J=\left[
\begin{matrix}
J_{11} & J_{12}\\
J_{21} & J_{22}	
\end{matrix}
\right],
$$
where
\begin{align*}
&J_{11}=\frac{p_1^{6}+3\,p_1^{5}p_2+3\,p_1^{4}p_2^{2}+\left(p_2^{3}+2\,k_1p_2\right)p_1^{3}-6\,k_1p_2 p_1^{2}c_1-6\,k_1 p_2^{2}p_1c_1-2\,c_1k_1p_2^{3}}{p_1^{3} \left(p_1+p_2\right)^{3}},\\
&J_{12}= \frac{k_1 \left(2\,c_1-p_1+p_2\right)}{\left(p_1+p_2\right)^{3}},\\
&J_{21}= \frac{k_2 \left(2\,c_2+p_1-p_2\right)}{\left(p_1+p_2\right)^{3}},\\
&J_{22}= \frac{p_2^{6}+3\,p_1p_2^{5}+3\,p_1^{2}p_2^{4}+\left(p_1^{3}+2\,k_2p_1\right)p_2^{3}-6\,k_2p_1p_2^{2}c_2-6\,k_2p_1^{2}p_2c_2-2\,c_2k_2 p_1^{3}}{p_2^{3} \left(p_1+p_2\right)^{3}}.
\end{align*}

Then the characteristic polynomial of $J$ is
$$CP(\lambda)=\lambda^2-\Tr(J)\lambda+\Det(J),$$
where $\Tr(J)=J_{11}+J_{22}$ and $\Det(J)=J_{11}J_{22}-J_{12}J_{21}$ are the trace and the determinant of $J$, respectively.
According to the Jury criterion \cite{Jury1976I}, the conditions for the local stability include:
\begin{enumerate}
	\item $CD_1^J\equiv CP(1)= 1-\Tr(J)+\Det(J)>0$,
	\item $CD_2^J\equiv CP(-1)= 1+\Tr(J)+\Det(J)>0$,
	\item $CD_3^J\equiv 1-\Det(J)>0$.
\end{enumerate}

\begin{remark}\label{rm:bif}
Furthermore, it is known that the discrete dynamic system may undergo a fold, period-doubling, or Neimark-Sacker bifurcation when the equilibrium loses its stability at $CD_1^J=0$, $CD_2^J=0$, or $CD_3^J=0$, respectively.
\end{remark}

\subsection{The special case of $c_1=c_2$}

If we set $c_1=c_2=c$ in \eqref{eq:equi-a12}, then the triangular decomposition method permits us to transform the equilibrium equations \eqref{eq:equi-a12} into the following three triangular sets.
\begin{align*}
	&\pset{T}_{21}=[p_1,~p_2],\\
	&\pset{T}_{22}=[p_1-3\,c,~p_2-3\,c],\\
	&\pset{T}_{23}=[p_1^{2}-cp_1-c^2,~p_2+p_1-c ].
\end{align*}
The zero of $\pset{T}_{21}$ is simply $(0,0)$.
From $\pset{T}_{23}$, we obtain two zeros\footnote{These zeros can also be obtained from $\pset{T}_{12}$ in \eqref{eq:ts12} by setting $c_1=c_2=c$.}
$$\left(\left(\frac{\sqrt{5}}{2}+\frac{1}{2}\right) c,\left(-\frac{\sqrt{5}}{2}+\frac{1}{2}\right) c\right),~~~\left(\left(-\frac{\sqrt{5}}{2}+\frac{1}{2}\right) c,\left(\frac{\sqrt{5}}{2}+\frac{1}{2}\right) c\right),$$
which are useless as the component $\left(-\frac{\sqrt{5}}{2}+\frac{1}{2}\right) c$ is negative. Therefore, the only non-vanishing equilibrium is $(3\,c,3\,c)$, which can be obtained from $\pset{T}_{22}$. 

%We investigate the local stability of $(3\,c,3\,c)$ in the sequel.

%The Jacobian matrix of the iteration map \eqref{eq:map-12} is
%$$J=
%\left[\begin{array}{cc}
%\frac{p_1^{6}+3 p_1^{5} p_2+3 p_1^{4} p_2^{2}+\left(p_2^{3}+2 k_1 p_2\right) p_1^{3}-6 c k_1 \,p_1^{2} p_2-6 c k_1 p_1 \,p_2^{2}-2 c k_1 \,p_2^{3}}{p_1^{3} \left(p_1+p_2\right)^{3}} & \frac{k_1 \left(2 c-p_1+p_2\right)}{\left(p_1+p_2\right)^{3}} 
%\\
% \frac{k_2 \left(2 c+p_1-p_2\right)}{\left(p_1+p_2\right)^{3}} & \frac{p_2^{6}+3 p_1 \,p_2^{5}+3 p_1^{2} p_2^{4}+\left(p_1^{3}+2 k_2 p_1\right) p_2^{3}-6 c k_2 p_1 \,p_2^{2}-6 c k_2 \,p_1^{2} p_2-2 c k_2 \,p_1^{3}}{p_2^{3} \left(p_1+p_2\right)^{3}} 
%\end{array}\right].
%$$

\begin{theorem}\label{thm:stability12-c}
	Let $\alpha=1/2$ and $c_1=c_2=c$. The unique non-vanishing equilibrium $(3\,c,3\,c)$ is locally stable if
$$c^2>\frac{2\, k_1+2\, k_2+\sqrt{4\, k_1^{2}-7\, k_1 k_2+4\, k_2^{2}}}{216}.$$
The system may undergo a period-doubling bifurcation when
$$c^2=\frac{2\, k_1+2\, k_2+\sqrt{4\, k_1^{2}-7\, k_1 k_2+4\, k_2^{2}}}{216}.$$
Furthermore, there exist no other bifurcations of the equilibrium.
\end{theorem}

\begin{proof}
Substituting $p_1=3\,c$ and $p_2=3\,c$ into $J$, we obtain that the Jacobian matrix at $(3\,c,3\,c)$ to be
$$J(3\,c,3\,c)=
\left[
\begin{matrix}
\frac{27\, c^{2}-k_1}{27\, c^{2}} & \frac{k_1}{216\, c^{2}} 
\\
 \frac{k_2}{216\, c^{2}} & \frac{27\, c^{2}-k_2}{27\, c^{2}} 
\end{matrix}
\right].
$$
Consequently,
\begin{align*}
	&\Tr(J)=\frac{54\, c^{2}-k_1-k_2}{27\, c^{2}}, \\
	&\Det(J)=\frac{5184\, c^{4}-192\, c^{2} k_1-192\, c^{2} k_2+7\, k_1 k_2}{5184\, c^{4}}.
\end{align*}

One can verify that the first condition for the local stability is always fulfilled since $k_1,k_2,c>0$ and
$$CD_1^J\equiv 1-\Tr(J)+\Det(J)=\frac{5\, k_1 k_2}{3888\, c^{4}}.$$	

The second condition is 
$$CD_2^J\equiv 1+\Tr(J)+\Det(J)=\frac{15552\, c^{4}+\left(-288\, k_1-288\, k_2\right) c^{2}+5\, k_1 k_2}{3888\, c^{4}}>0,$$
which means that
$$
15552\, c^{4}+\left(-288\, k_1-288\, k_2\right) c^{2}+5\, k_1 k_2>0,
$$
i.e.,
$$c^2>\frac{2\, k_1+2\, k_2+\sqrt{4\, k_1^{2}-7\, k_1 k_2+4\, k_2^{2}}}{216}~~\text{or}~~
c^2<\frac{2\, k_1+2\, k_2-\sqrt{4\, k_1^{2}-7\, k_1 k_2+4\, k_2^{2}}}{216}.$$

The third condition is
$$CD_3^J\equiv 1-\Det(J)=\frac{\left(144\, k_1+144\, k_2\right) c^{2}-5\, k_1 k_2}{3888\, c^{4}}>0,$$
which implies that
$$\left(144\, k_1+144\, k_2\right) c^{2}-5\, k_1 k_2>0,$$
i.e.,
$$c^2>\frac{5\, k_1 k_2}{144 \left(k_1+k_2\right)}.$$

Furthermore, it can be proved that 
$$\frac{2\, k_1+2\, k_2-\sqrt{4\, k_1^{2}-7\, k_1 k_2+4\, k_2^{2}}}{216}<
\frac{5\, k_1 k_2}{144\, \left(k_1+k_2\right)}<\frac{2\, k_1+2\, k_2+\sqrt{4\, k_1^{2}-7\, k_1 k_2+4\, k_2^{2}}}{216}.$$
Accordingly, the equilibrium is locally stable if
$$c^2>\frac{2\, k_1+2\, k_2+\sqrt{4\, k_1^{2}-7\, k_1 k_2+4\, k_2^{2}}}{216}.$$
The rest of the proof follows immediately from Remark \ref{rm:bif}.
\end{proof}

Figure \ref{fig:c1c2-12} depicts two 2-dimensional cross-sections of the stability region reported in Theorem \ref{thm:stability12-c}. It is observed that an increase in the marginal cost $c$ or a decrease in the adjustment speeds $k_1,k_2$ has an effect of stabilizing the unique non-vanishing equilibrium.

\begin{figure}[htbp]
  \centering
  \subfloat[$k_2=1/10$]{\includegraphics[width=0.35\textwidth]{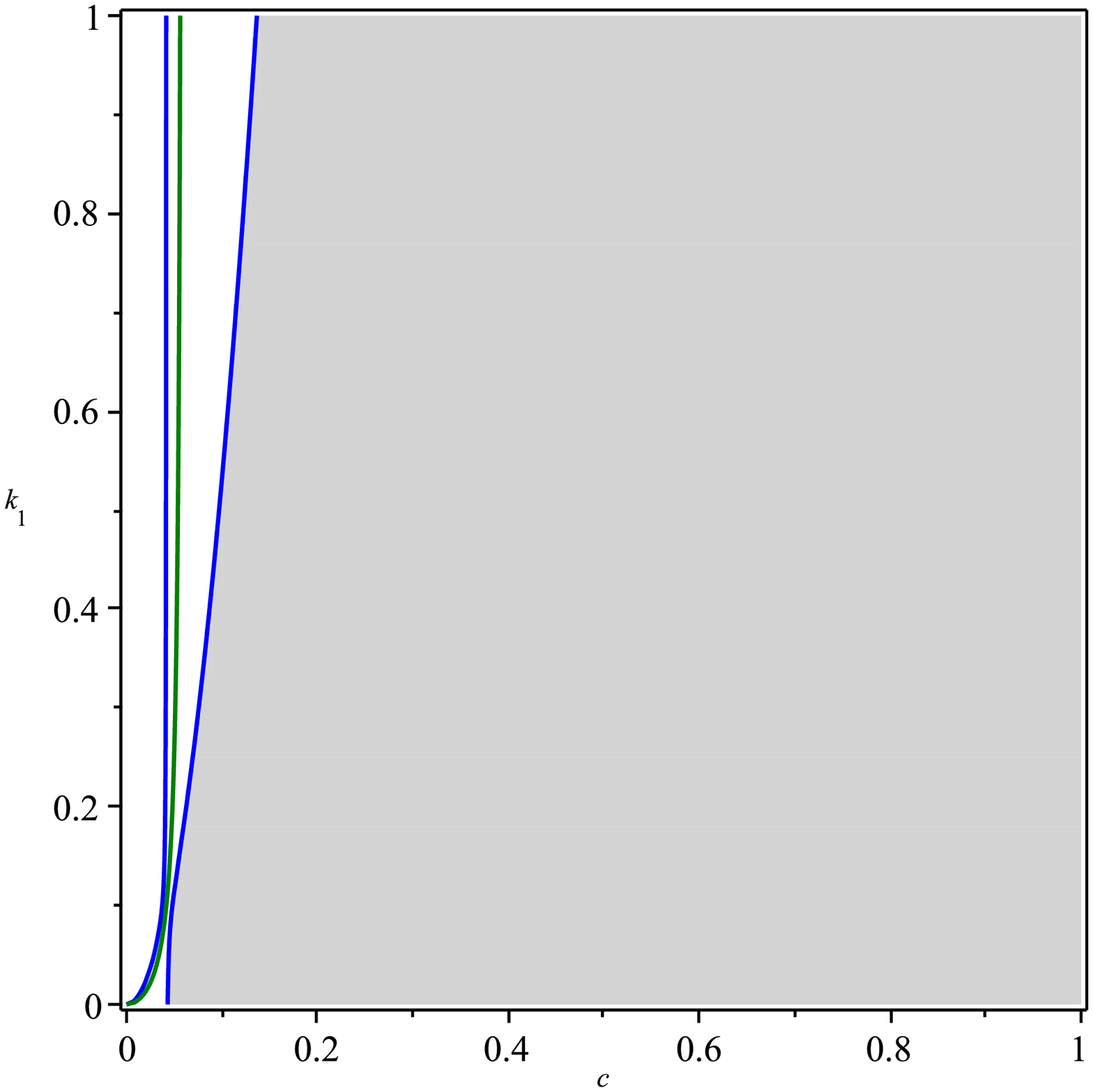}} 
  \subfloat[$c=1/3$]{\includegraphics[width=0.35\textwidth]{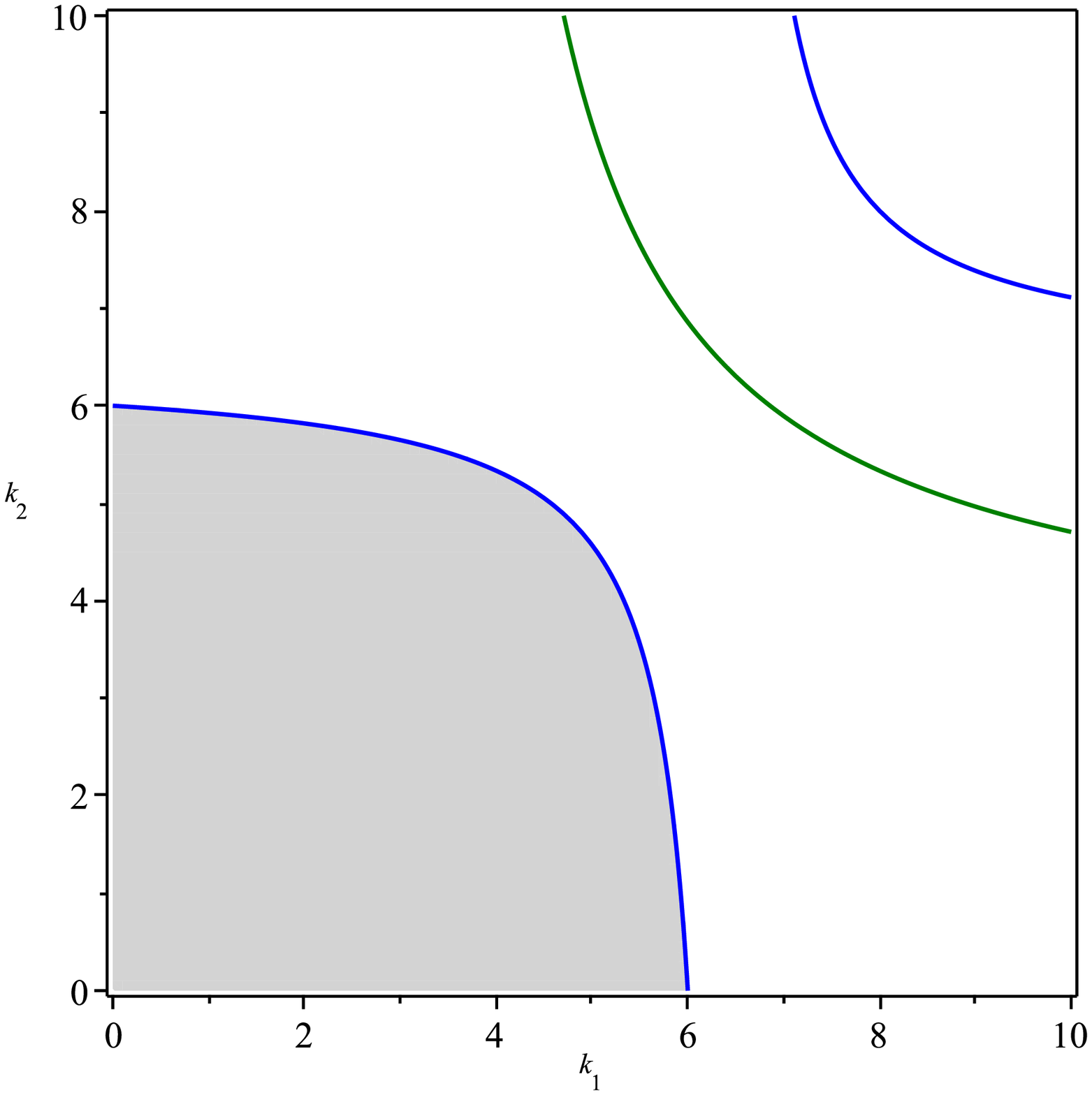}} \\
  \caption{The 2-dimensional cross-sections of the stability region of the considered model with $\alpha=1/2$ and $c_1=c_2=c$. The curves of $CD_2^J=0$ and $CD_3^J=0$ are marked in blue and green, respectively.}
 \label{fig:c1c2-12}
\end{figure}

%\memo{Neimark-Sacker can not occur}

\subsection{The general case}

If $c_1\neq c_2$, then the analytical expression of the unique non-vanishing equilibrium would be quite complicated. Thus, the proof of Theorem \ref{thm:stability12-c} can not work since it is impossible to substitute the analytical expression of the equilibrium into the Jacobian matrix and obtain a neat result. Concerning the bifurcation analysis, we need to determine the conditions on the parameters that $CD_1^J=0$, $CD_2^J=0$, and $CD_3^J=0$ are satisfied at the non-vanishing equilibrium. For this purpose, the following notation is required.

\begin{definition}
Let $$A=\sum_{i=0}^ma_i\,x^i,~~B=\sum_{j=0}^lb_j\,x^j$$ be two
univariate polynomials in $x$ with coefficients $a_i,b_j$, and $a_m,b_l\neq 0$. The
determinant
\begin{equation*}\label{eq:sylmat}
 \begin{array}{c@{\hspace{-5pt}}l}
 \left|\begin{array}{cccccc}
a_m & a_{m-1}& \cdots   & a_0   &        &       \\
           & \ddots   & \ddots&    \ddots    &\ddots&   \\
         &          & a_m   & a_{m-1}&\cdots& a_0 \\ [5pt]
b_l & b_{l-1}& \cdots   &  b_0 &    &         \\
            & \ddots   &\ddots &   \ddots     &\ddots&       \\
         &   &    b_{l}     & b_{l-1} & \cdots &  b_0
\end{array}\right|
& \begin{array}{l}\left.\rule{0mm}{8mm}\right\}l\\
\\\left.\rule{0mm}{8mm}\right\}m
\end{array}%[-5pt]
\end{array}
\end{equation*}
is called the \emph{Sylvester resultant} (or simply
\emph{resultant}) of $A$ and $B$ with respect to $x$, and denoted by $\res(A,B,x)$. 	
\end{definition}

The following lemma reveals the main property of the resultant, which can also be found in \cite{Mishra1993A}.

\begin{lemma}\label{lem:res-com}
Let $A$ and $B$ be two univariate polynomials in $x$.
There exist two polynomials $F$ and $G$ in $x$ such that 
$$FA+GB=\res(A,B,x).$$
Furthermore, $A$ and $B$ have common zeros in the field of complex numbers if and only if $\res(A, B)=0$.
\end{lemma}

For a triangular set $\pset{T}=[T_1(x),T_2(x,y)]$ and a polynomial $H(x,y)$, we define 
$$\res(H,\pset{T})\equiv \res(\res(H,T_2,y),T_1(x),x).$$
By Lemma \ref{lem:res-com}, if $T_1=0$ and $T_2=0$ (or simply denoted as $\pset{T}=0$), then one knows that $H=0$ implies $\res(H,\pset{T})=0$, which means $\res(H,\pset{T})=0$ is a necessary condition for $H=0$. Consequently, the following proposition is acquired. It should be emphasized that Proposition \ref{prop:bifur-12-k} only reports the results for the case of $k_1=k_2$ because the conditions for $k_1\neq k_2$ are too long to list in this paper due to space limitations. However, readers can see that the idea of the proof also works for $k_1\neq k_2$ and can derive the complete conditions themself.

\begin{proposition}\label{prop:bifur-12-k}
	Let $\alpha=1/2$ and $k_1=k_2=k$. The system may undergo a period-doubling bifurcation when $R_1=0$ and a Neimark-Sacker bifurcation when $R_2=0$, where $R_1$ and $R_2$ are given in Appendix.
\end{proposition}
\begin{proof}
It should be noted that the resultant is feasible only for polynomials.
For $CD_1^J$, we consider its numerator $\numer(CD_1^J)$. Then one can obtain that
$$
\res(\numer(CD_1^J),\pset{T}_{12})=81\, k^{6} c_1^{18} c_2^{6} \left(c_1+c_2\right) \left(32 c_1^{2}+61 c_1 c_2+32 c_2^{2}\right).$$
Since $c_1>0$, $c_2>0$, and $k>0$, it is impossible that $\res(\numer(CD_1^J),\pset{T}_{12})=0$ or $CD_1^J=0$ provided that $\pset{T}_{12}=0$. Hence, the equilibrium can not lose its stability through a fold bifurcation. Furthermore, we have
\begin{align*}
&\res(\numer(CD_2^J),\pset{T}_{12})=-729\,c_1^{32} c_2^8 (c_1 + c_2) R_1,\\
&\res(\numer(CD_3^J),\pset{T}_{12})=729\,k^3c_1^{32} c_2^8 (c_1 + c_2) R_2,
\end{align*}
which will vanish only if $R_1=0$ and $R_2=0$, respectively.
Consequently, the system may undergo a period-doubling bifurcation when $R_1=0$ and a Neimark-Sacker bifurcation when $R_2=0$.	 

\end{proof}

By Proposition \ref{prop:unique-sol-12}, there exists only one equilibrium $(p_1,p_2)$ with $p_1>0$ and $p_2>0$ although its analytical expression is complicated. To explore the local stability, we need to determine the signs of $CD_1^J$, $CD_2^J$, and $CD_3^J$ at this equilibrium without using its closed form. It should be noted that $CD_1^J$, $CD_2^J$, and $CD_3^J$ are rational functions. Suppose that 
$$CD_i^J=\frac{\numer(CD_i^J)}{\denom(CD_i^J)},$$
where $\numer(\cdot)$ and $\denom(\cdot)$ denote the numerator and the denominator, respectively. Then the sign of $CD_i^J$ is the same as that of $\numer(CD_i^J)\cdot \denom(CD_i^J)$ if $\denom(CD_i^J)\neq 0$. One could compute that
$$\res(\numer(CD_1^J)\cdot \denom(CD_1^J),\pset{T}_{12})=-1594323\,k^6 c_1^{50} c_2^{17} (c_1 + c_2)^6 (32\, c_1^2 + 61\,c_1 c_2 + 32\,c_2^2),$$
$$\res(\numer(CD_2^J)\cdot \denom(CD_2^J),\pset{T}_{12})=129140163\,c_1^{70} c_2^{22} (c_1 + c_2)^6 R_1,$$
and 
$$\res(\numer(CD_3^J)\cdot \denom(CD_3^J),\pset{T}_{12})=-129140163\, k^3 c_1^{70} c_2^{22} (c_1 + c_2)^6 R_2.$$
We should emphasize that the sign of $\res(\numer(CD_i^J)\cdot \denom(CD_i^J),\pset{T}_{12})$ may not be the same as $\numer(CD_i^J)\cdot \denom(CD_i^J)$ or $CD_i^J$. However, it is known that $\res(\numer(CD_i^J)\cdot \denom(CD_i^J),\pset{T}_{12})$ involves only the parameters and its zeros divide the parameter space into several regions. In each region, the sign of $CD_i^J$ is invariant. Consequently, we just need to select one sample point from each region and identify the sign of $CD_i^J$ at the selected sample point. The selection of sample points might be extremely complicated in general and could be automated using, e.g., the PCAD method \cite{Collins1991P}. 

In Table \ref{tab:sample-12}, we list all the selected sample points and the corresponding information on whether the non-vanishing equilibrium is stable, i.e., whether $CD_1^J>0$, $CD_2^J>0$, and $CD_3^J>0$ are simultaneously satisfied. Moreover, Table \ref{tab:sample-12} displays the signs of $R_1$ and $R_2$ at these sample points. One can observe that the equilibrium is stable if $R_1>0$ and $R_2>0$, and vice versa. It should be mentioned that the calculations involved in Table \ref{tab:sample-12} are exact and rigorous. That is, the computational results provide theoretical foundations for a systematic analysis of the local stability. Therefore, we acquire the following theorem.

%Then identifying the local stability of this equilibrium can be transformed into analyzing the following semi-algebraic system.
%\begin{equation}\label{eq:semi-k1k2-12}
%\left\{
%\begin{split}
%	&-2p_2 p_1^{2}+\left(p_2^{2}+2p_2 p_1\right) c_1=0,\\
%	&-2p_1 p_2^{2}+\left(p_1^{2}+2p_1 p_2\right) c_2=0,\\
%	&CD_1^J^*>0,~CD_2^J^*>0,~CD_3^J^*>0,\\
%	&p_1>0,~p_2>0,~c_1>0,~c_2>0,~k>0.
%\end{split}\right.
%\end{equation}
%It is worth noting that the method given in \cite{} can handle systems involving only polynomials. Thus, we need to transform the conditions $CP(1)>0$, $CP(-1)>0$ and $1-\Det(J)>0$ (involving rational functions) into equivalent conditions involving only polynomials, which are denoted by $CP(1)^*$, $CP(-1)^*$ and $(1-\Det(J))^*$, respectively. For instance, suppose that $1-\Det(J)=A/B$, then $1-\Det(J)>0$ is equivalent to $AB>0$ provided that $B\neq 0$. Hence, $(1-\Det(J))^*=AB$. 
%
% 
%As aforementioned, the exploration of the local stability can be transformed into the determination of the parameter conditions of the existence of real solutions of system \eqref{eq:semi-k1k2-12}. We can solve this problem using the symbolic approach given in \cite{}. Here, we avoid involving tedious details but just present the final results.
%

\begin{theorem}\label{thm:stable-12-k1k2}
If $k_1=k_2=k$, the unique non-vanishing equilibrium $(p_1,p_2)$ with $p_1>0$ and $p_2>0$ is locally stable if
$R_1>0$ and $R_2>0$, where $R_1$ and $R_2$ can be found in Appendix.
\end{theorem}

%In Figure \ref{fig:k1k2-12}, we display the region of the local stability given in Theorem \ref{thm:stable-12-k1k2} with 2-dimensional cross-sections. \memo{observations}

\begin{table}[htbp]
	\centering 
	\caption{Selected Sample Points in $\{(c_1,c_2,k)\,|\,c_1> 0, c_2>0, k>0\}$ for $\alpha=1/2$}
	\label{tab:sample-12} 
	\begin{tabular}{|l|c|c|c||l|c|c|c|}  
		\hline  
		$(c_1,c_2,k)$ & stable & $R_1$ & $R_2$ &$(c_1,c_2,k)$ &  stable & $R_1$ & $R_2$ \\ \hline
$(1, 1/4, 1)$& yes & $+$ & $+$ & $(1, 5/16, 1)$& yes & $+$ & $+$ \\ \hline
$(1, 1/4, 7)$& no & $-$ & $+$ & $(1, 5/16, 10)$& no & $-$ & $+$ \\ \hline
$(1, 1/4, 29)$& no & $-$ & $-$ & $(1, 5/16, 30)$& no & $-$ & $-$ \\ \hline
$(1, 1/4, 51)$& no & $+$ & $-$ & $(1, 5/16, 51)$& no & $+$ & $-$ \\ \hline

$(1, 1/2, 1)$& yes & $+$ & $+$ & $([1, 7/8, 1)$& yes & $+$ & $+$ \\ \hline
$(1, 1/2, 18)$& no & $-$ & $+$ & $(1, 7/8, 38)$& no & $-$ & $+$ \\ \hline
$(1, 1/2, 35)$& no & $-$ & $-$ &$(1, 7/8, 51)$& no & $-$ & $-$ \\ \hline
$(1, 1/2, 53)$& no & $+$ & $-$ &$(1, 7/8, 65)$& no & $+$ & $-$ \\ \hline

$(1, 9/8, 1)$& yes & $+$ & $+$ & $(1, 2, 1)$& yes & $+$ & $+$ \\ \hline
$(1, 9/8, 49)$& no & $-$ & $+$ & $(1, 2, 70)$& no & $-$ & $+$ \\ \hline
$(1, 9/8, 66)$& no & $-$ & $-$ & $(1, 2, 140)$& no & $-$ & $-$ \\ \hline
$(1, 9/8, 83)$& no & $+$ & $-$ & $(1, 2, 209)$& no & $+$ & $-$ \\ \hline

$(1, 3, 1)$& yes & $+$ & $+$ & $(1, 4, 1)$& yes & $+$ & $+$ \\ \hline
$(1, 3, 91)$& no & $-$ & $+$ & $(1, 4, 112)$& no & $-$ & $+$ \\ \hline
$(1, 3, 272)$& no & $-$ & $-$ & $(1, 4, 462)$& no & $-$ & $-$ \\ \hline
$(1, 3, 453)$& no & $+$ & $-$ & $(1, 4, 811)$& no & $+$ & $-$ \\ \hline
	\end{tabular}
\end{table}

%\subsection{$c_1=1/10$, $c_2=15/100$}
%
%$[R[1] < 0, 0 < R[2]]$
%where  
%\begin{align*}
%&R[1] = 187680000*k1^3*k2^3-95716800*k1^3*k2^2-58180800*k1^2*k2^3-8686584*k1^3*k2-7254264*k1^2*k2^2-1495744*k1*k2^3-3645*k1^3-194400*k1^2*k2-113400*k1*k2^2-2430*k2^3,\\
%&R[2] = 7507200000000*k1^3*k2^3-7657344000000*k1^3*k2^2-4654464000000*k1^2*k2^3-1389853440000*k1^3*k2+4384886400000*k1^2*k2^2-239319040000*k1*k2^3-1166400000*k1^3+803416377600*k1^2*k2+226086777600*k1*k2^2-777600000*k2^3+659016000*k1^2+42193368192*k1*k2+711504000*k2^2+44916120*k1+162946080*k2+177147	
%\end{align*}

\section{$\alpha=1/3$}

If $\alpha=1/3$, then $\beta=1/2$. We have the iteration map 
\begin{equation}\label{eq:map-13}
\left\{
\begin{split}
&p_1(t+1)=p_1(t)+k_1\frac{-p_1(t)\sqrt{p_1(t)p_2(t)}+\left(2\,p_2(t)+3\sqrt{p_1(t) p_2(t)}\right) c_1}{2\,p_1^{2}(t) \left(\sqrt{p_1(t)}+\sqrt{p_2(t)}\right)^{2}},\\
&p_2(t+1)=p_2(t)+k_2\frac{-p_2(t)\sqrt{p_1(t)p_2(t)}+\left(2\,p_1(t)+3\sqrt{p_1(t) p_2(t)}\right) c_2}{2\,p_2^{2}(t) \left(\sqrt{p_1(t)}+\sqrt{p_2(t)}\right)^{2}}.
\end{split}\right.
\end{equation}

By setting $p_1(t+1)=p_1(t)=p_1$ and $p_2(t+1)=p_2(t)=p_2$, one can obtain the equations of the equilibrium
\begin{equation*}
\left\{
\begin{split}
&-p_1\sqrt{p_1p_2}+\left(2\,p_2+3\sqrt{p_1 p_2}\right) c_1=0,\\	
&-p_2\sqrt{p_1p_2}+\left(2\,p_1+3\sqrt{p_1 p_2}\right) c_2=0.
\end{split}\right.
\end{equation*}
Denote $\sqrt{p_1}=x$ and $\sqrt{p_2}=y$. The above equations become
\begin{equation}\label{eq:fixp-13-xy}
\left\{
\begin{split}
&-x^3y+(2\,y^2+3\,xy)c_1=0,\\
&-y^3x+(2\,x^2+3\,xy)c_2=0.
\end{split}
\right.
\end{equation}

Using the triangular decomposition method, we decompose the solutions of system \eqref{eq:fixp-13-xy} into zeros of the following two triangular sets.
\begin{align*}
	&\pset{T}_{31}=\left[x,~y\right],\\
	&\pset{T}_{32}=\left[x^{8}-9\,c_1 x^{6}+27\,c_1^{2} x^{4}+(-27\,c_1^{3}-12\,c_1^{2}c_2) x^{2}+20\,c_1^{3}c_2,~2\,c_1 y-x^{3}+3\,c_1 x \right].
\end{align*}

Evidently, $\pset{T}_{31}$ is corresponding to the origin $(0,0)$.
Therefore, the identification of the number of non-vanishing equilibria can be transformed into the determination of the number of real solutions of the following semi-algebraic system.

\begin{equation*}
\left\{
	\begin{split}
&x^{8}-9\,c_1 x^{6}+27\,c_1^{2} x^{4}+(-27\,c_1^{3}-12\,c_1^{2}c_2) x^{2}+20\,c_1^{3}c_2=0,\\
&2\,c_1 y-x^{3}+3\,c_1 x=0,\\
&x>0,~y>0.
	\end{split}
\right.
\end{equation*}

Using the algebraic approach by Li and Wang \cite{Li2014C}, we know that the above system has one unique real solution for any parameter values of $c_1,c_2>0$, which implies the following proposition.

\begin{proposition}\label{prop:unique-sol-13}
	Let $\alpha=1/3$. The iteration map \eqref{eq:map-13} possesses one unique equilibrium $(p_1,p_2)$ with $p_1>0$ and $p_2>0$.
\end{proposition}

To investigate the local stability of the equilibrium, we consider the Jacobian matrix 
$$
M=\left[
\begin{matrix}
M_{11} & M_{12}\\
M_{21} & M_{22}	
\end{matrix}
\right],
$$
where
\begin{align*}
&M_{11}= \frac{12\,p_1^{\frac{9}{2}} \sqrt{p_2}+4\,p_1^{\frac{7}{2}}p_2^{\frac{3}{2}}-15\,c_1k_1 p_1^{\frac{3}{2}} \sqrt{p_2}-8\,c_1k_1 p_2^{\frac{3}{2}} \sqrt{p_1}+3\,k_1 p_1^{\frac{5}{2}} \sqrt{p_2}+4\,p_1^{5}+12\,p_1^{4}p_2-21\,c_1k_1p_1p_2+k_1 p_1^{2}p_2}{4\,p_1^{\frac{7}{2}} \left(\sqrt{p_1}+\sqrt{p_2}\right)^{3}},\\
&M_{12}= \frac{k_1 \left(\sqrt{p_2}\,p_1^{\frac{3}{2}}-p_1^{2}+c_1 \sqrt{p_2} \sqrt{p_1}+3\,p_1c_1\right)}{4\,p_1^{2} \left(\sqrt{p_1}+\sqrt{p_2}\right)^{3} \sqrt{p_2}},\\
&M_{21}= \frac{k_2 \left(\sqrt{p_1}\,p_2^{\frac{3}{2}}+c_2 \sqrt{p_2} \sqrt{p_1}+3\,p_2c_2-p_2^{2}\right)}{4\,p_2^{2} \left(\sqrt{p_1}+\sqrt{p_2}\right)^{3} \sqrt{p_1}},\\
&M_{22}= \frac{4\,p_1^{\frac{3}{2}}p_2^{\frac{7}{2}}+12\,p_2^{\frac{9}{2}} \sqrt{p_1}-8\,c_2k_2p_1^{\frac{3}{2}} \sqrt{p_2}-15\,c_2k_2 p_2^{\frac{3}{2}} \sqrt{p_1}+3\,k_2 p_2^{\frac{5}{2}} \sqrt{p_1}+12\,p_1 \,p_2^{4}+4\,p_2^{5}-21\,c_2k_2p_1p_2+k_2p_1 p_2^{2}}{4p_2^{\frac{7}{2}} \left(\sqrt{p_1}+\sqrt{p_2}\right)^{3}}.\\	
\end{align*}

As in Section 3, we denote
\begin{align*}
&CD_1^M \equiv 1-\Tr(M)+\Det(M),\\
&CD_2^M \equiv 1+\Tr(M)+\Det(M),\\
&CD_3^M \equiv 1-\Det(M).
\end{align*}

%\subsection{$c_2=c_1$, $k_2=k_1$}
%
%$[0 < R[1]]$
%where
%$$R[1] = 2000*c_1^2-7*k1$$

\subsection{The special case of $c_1=c_2$}

If we set $c_1=c_2=c$, then the triangular decomposition method permits us to transform the equilibrium equations \eqref{eq:fixp-13-xy} into the following triangular sets.
\begin{align*}
	&\pset{T}_{41}=[x,y],\\
	&\pset{T}_{42}=[x^{2}-c,y+x],\\
	&\pset{T}_{43}=[x^{2}-5\,c,y-x],\\
	&\pset{T}_{44}=[x^{4}-3\, c \,x^{2}+4\, c^{2},2\, c y-x^{3}+3\,c x ].
\end{align*}

Obviously, the zeros of $\pset{T}_{41}$ and $\pset{T}_{42}$ are economically uninteresting.
Moreover, all the roots of $x^{4}-3\, c \,x^{2}+4\, c^{2}$ of $\pset{T}_{44}$, i.e.,
$$
\frac{\sqrt{2 \,\sqrt{7}c\, \mathrm{i}+6\, c}}{2},-\frac{\sqrt{2 \,\sqrt{7}c\, \mathrm{i}+6\, c}}{2},\frac{\sqrt{-2 \,\sqrt{7}c\, \mathrm{i}+6\, c}}{2},-\frac{\sqrt{-2 \,\sqrt{7}c\, \mathrm{i}+6\, c}}{2},
$$ 
are imaginary and also not of our concern. There exists only one non-vanishing equilibrium $(p_1,p_2)=(5\,c,5\,c)$, which is corresponding to the branch $\pset{T}_{43}$.

Substituting $p_1=5\,c$ and $p_2=5\,c$ into $M$, we obtain the Jacobian matrix at the equilibrium $(5\,c,5\,c)$ to be
$$M(5\,c,5\,c)=
\left[
\begin{array}{cc}
\frac{500\, c^{2}-3\, k_1}{500\, c^{2}} & \frac{k_1}{1000\, c^{2}} 
\\
 \frac{k_2}{1000\, c^{2}} & \frac{500\, c^{2}-3\, k_2}{500\, c^{2}} 
\end{array}
\right].
$$
Hence, 
\begin{align*}
	&\Tr(M)=\frac{1000\, c^{2}-3\, k_1-3 k_2}{500\, c^{2}}, \\
	&\Det(M)=\frac{200000\, c^{4}-1200\, c^{2} k_1-1200\, c^{2} k_2+7\, k_1 k_2}{200000\, c^{4}}.
\end{align*}

\begin{theorem}\label{thm:stability13-c}
	Let $\alpha=1/3$ and $c_1=c_2=c$. The unique non-vanishing equilibrium $(5\,c,5\,c)$ is locally stable if
$$c^2>\frac{3\, k_1+3\, k_2+\sqrt{9\, k_1^{2}-17\, k_1 k_2+9\, k_2^{2}}}{2000}.$$
The system may undergo a period-doubling bifurcation when
$$c^2=\frac{3\, k_1+3\, k_2+\sqrt{9\, k_1^{2}-17\, k_1 k_2+9\, k_2^{2}}}{2000}.$$
Furthermore, there exist no other bifurcations of the equilibrium.
\end{theorem}

\begin{proof}
The first condition for the local stability is always fulfilled since
$$CD_1^M\equiv 1-\Tr(M)+\Det(M)=\frac{7\, k_1 k_2}{200000 c^{4}}.$$	

The second condition should be 
$$CD_2^M\equiv 1+\Tr(M)+\Det(M)=\frac{800000\, c^{4}+\left(-2400\, k_1-2400\, k_2\right) c^{2}+7\, k_1 k_2}{200000\, c^{4}}>0,$$
which implies that
$$
{800000\, c^{4}+\left(-2400\, k_1-2400\, k_2\right) c^{2}+7\, k_1 k_2}>0,
$$
i.e.,
$$c^2>\frac{3\, k_1+3\, k_2+\sqrt{9\, k_1^{2}-17\, k_1 k_2+9\, k_2^{2}}}{2000}~~\text{or}~~
c^2<\frac{3\, k_1+3\, k_2-\sqrt{9\, k_1^{2}-17\, k_1 k_2+9\, k_2^{2}}}{2000}.$$

The third condition should be
$$CD_3^M\equiv 1-\Det(M)=\frac{\left(1200\, k_1+1200\, k_2\right) c^{2}-7\, k_1 k_2}{200000\, c^{4}}>0,$$
from which we have
$${\left(1200\, k_1+1200\, k_2\right) c^{2}-7\, k_1 k_2}>0,$$
i.e.,
$$c^2>\frac{7\, k_1 k_2}{1200\, \left(k_1+k_2\right)}.$$

It can be proved that
$$\frac{3\, k_1+3\, k_2-\sqrt{9\, k_1^{2}-17\, k_1 k_2+9\, k_2^{2}}}{2000}<\frac{7\, k_1 k_2}{1200 \left(k_1+k_2\right)}<\frac{3\, k_1+3\, k_2+\sqrt{9\, k_1^{2}-17\, k_1 k_2+9\, k_2^{2}}}{2000}.$$
Therefore, the equilibrium is locally stable if
$$c^2>\frac{3\, k_1+3\, k_2+\sqrt{9\, k_1^{2}-17\, k_1 k_2+9\, k_2^{2}}}{2000}.$$
The rest of the proof follows from Remark \ref{rm:bif}.
\end{proof}

In Figure \ref{fig:c1c2-13}, we show two 2-dimensional cross-sections of the stability region reported in Theorem \ref{thm:stability13-c}. One can see that the equilibrium may lose its stability if the adjustment speeds $k_1,k_2$ are large enough or the marginal cost $c$ is small enough.

%
%$[0 < R[1], 0 < R[2]]$
%where                                   
%\begin{align*}
%&R[1] = 1200*c_1^2*k1+1200*c_1^2*k2-7*k1*k2,\\
%&R[2] = 800000*c_1^4-2400*c_1^2*k1-2400*c_1^2*k2+7*k1*k2	
%\end{align*}

\begin{figure}[htbp]
  \centering
  \subfloat[$k_2=1/10$]{\includegraphics[width=0.35\textwidth]{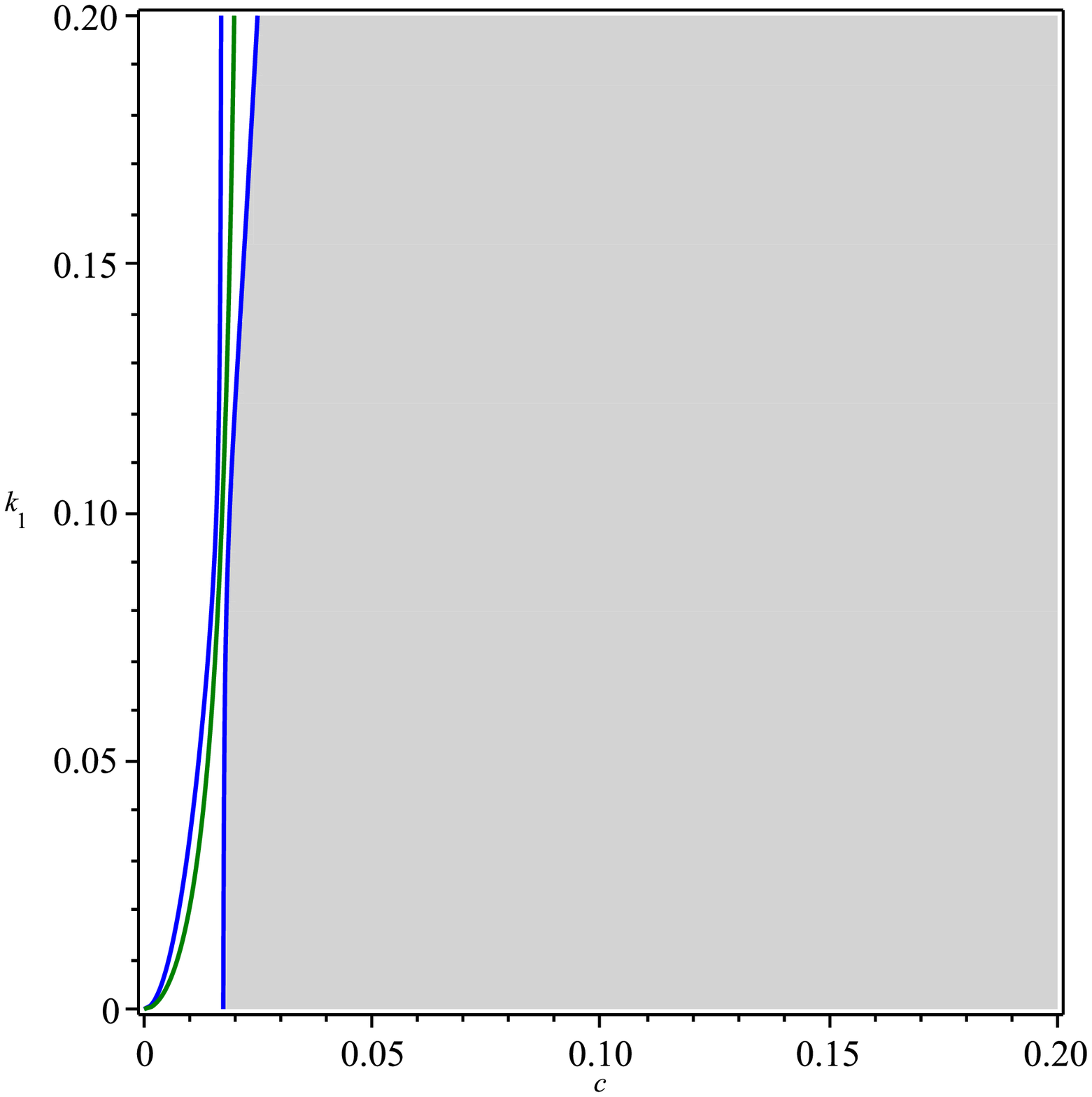}} 
  \subfloat[$c=1/10$]{\includegraphics[width=0.35\textwidth]{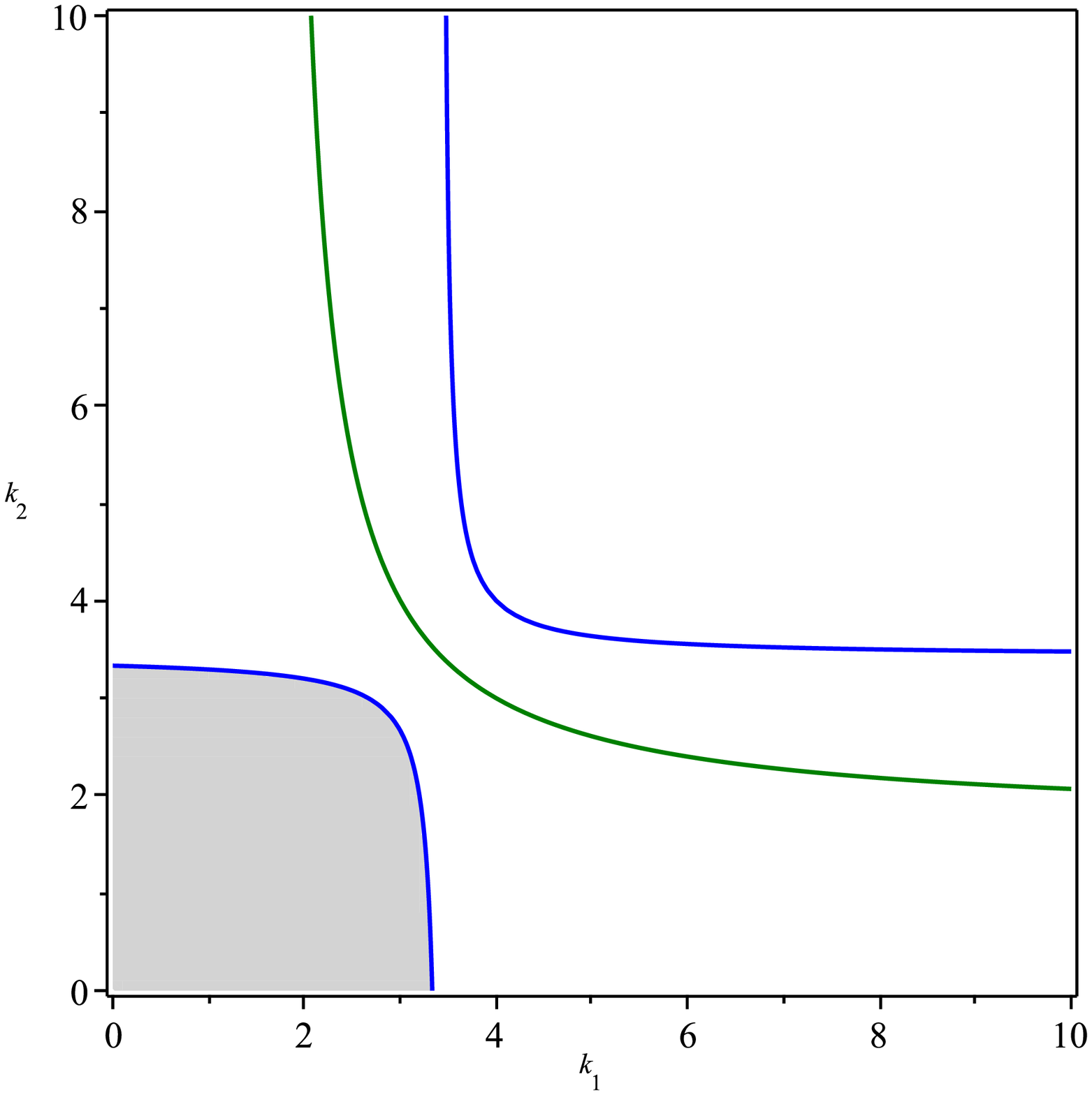}} \\
  \caption{The 2-dimensional cross-sections of the stability region of the considered model with $\alpha=1/3$ and $c_1=c_2=c$. The curves of $CD_2^M=0$ and $CD_3^M=0$ are marked in blue and green, respectively.}
 \label{fig:c1c2-13}
\end{figure}

%\memo{Neimark-Sacker can not occur}

\subsection{The general case}

As in Section 3.2, we set $k_1=k_2=k$. {We should mention that the method employed in this section also works for the case of $k_1\neq k_2$. However, the conditions for $k_1\neq k_2$ are tedious and not reported in this section due to space limitations. Interested readers can use our method to compute the complete conditions themself.} The case of $c_1=c_2$ has been explored in Section 4.1, hence we suppose that $c_1\neq c_2$ in what follows. The bifurcations are analyzed in the following proposition.

\begin{proposition}\label{prop:bifur-13-k}
	Let $\alpha=1/3$, $k_1=k_2=k$ and $c_1\neq c_2$. The iteration map \eqref{eq:map-13} may undergo a period-doubling bifurcation when $R_3=0$ and a Neimark-Sacker bifurcation when $R_4=0$, where $R_3$ and $R_4$ are given in Appendix.
\end{proposition}

\begin{proof}
Computing the resultant of $\numer(CD_1^M)$ with respect to $\pset{T}_{32}$, one obtains
\begin{align*}
\res(\numer(CD_1^M),\pset{T}_{32})=879609302220800000\, k^{16} c_1^{51} c_2^{11} \left(c_1-c_2\right)^{2} \left(2187\, c_1^{2}-4031\, c_1 c_2+2187\, c_2^{2}\right)^{2}.
\end{align*}
It is evident that
\begin{equation*}
2187\,c_1^{2} - 4031\,c_1c_2 + 2187\,c_2^{2}=2187(c_1-c_2)^2+343\,c_1c_2>0.	
\end{equation*}
Therefore, $\res(\numer(CD_1^M),\pset{T}_{32})\neq 0$, which means that $CD_1^M\neq 0$ at the unique non-vanishing equilibrium. Hence, there exist no fold bifurcations in map \eqref{eq:map-13}. Furthermore, we have
\begin{align*}
\res(\numer(CD_2^J),\pset{T}_{32})=99035203142830421991929937920000000\, c_1^{101} c_2^{13}\left(c_1-c_2\right)^{2} R_3^2,
\end{align*}
and 
\begin{align*}
\res(\numer(CD_3^J),\pset{T}_{32})=99035203142830421991929937920000000\,k^8 c_1^{101} c_2^{13} (c_1-c_2)^{10} R_4^2.
\end{align*}
Consequently, a period-doubling bifurcation may occur when $R_3=0$, while a Neimark-Sacker bifurcation may take place when $R_4=0$.	
\end{proof}

To investigate the local stability, we need to consider $\numer(CD_i^J)\cdot \denom(CD_i^J)$ and compute its resultant with respect to $\pset{T}_{32}$. Then it is obtained that
\begin{align*}
&\res(\numer(CD_1^J)\cdot \denom(CD_1^J),\pset{T}_{32})=\\
&\qquad5708990770823839524233143877797980545530986496\cdot 10^{20}\\
&\qquad\cdot k^{16} c_1^{156} c_2^{36} 
(c_1-c_2)^{12} 
(2187\,c_1^{2} 
- 4031\,c_1c_2 
+ 2187\,c_2^{2})^{2},\\

&\res((\numer(CD_2^J)\cdot \denom(CD_2^J),\pset{T}_{32})=\\
&\qquad 6582018229284824168619876730229402019930943462534319453394436096\cdot 10^{24}\\
&\qquad\cdot c_1^{218} c_2^{42} (c_1-c_2)^{10} R_3^2,\\

&\res((\numer(CD_3^J)\cdot \denom(CD_3^J),\pset{T}_{32})=\\
&\qquad6582018229284824168619876730229402019930943462534319453394436096\cdot 10^{24}\\
&\qquad\cdot k^{8} c_1^{218} c_2^{42} (c_1-c_2)^{10} R_4^2.
\end{align*}

These $\res(\numer(CD_i^J)\cdot \denom(CD_i^J),\pset{T}_{32})$ involve only the parameters and their zeros divide the parameter set $\{(c_1,c_2,k)\,|\,c_1> 0, c_2>0, k>0\}$ into several regions. In each region, the signs of $CD_1^M$, $CD_2^M$, and $CD_3^M$ are fixed and can be identified by checking at a selected sample point. In Table \ref{tab:sample-13}, we list the $40$ selected sample points and the signs of $R_3$, $R_4$ at these sample points. Moreover, Table \ref{tab:sample-13} provides the information on whether the non-vanishing equilibrium is stable, i.e., whether the stability conditions $CD_1^M>0$, $CD_2^M>0$ and $CD_3^M>0$ are satisfied simultaneously. Interested readers may check the correctness of Table \ref{tab:sample-13} themselves. Based on a series of computations, we acquire the following theorem.

\begin{theorem}\label{thm:stable-13-k1k2}
Let $k_1=k_2=k$ and $c_1\neq c_2$. The unique non-vanishing equilibrium of map \eqref{eq:map-13} is locally stable if one of the following conditions is satisfied:
\begin{enumerate}
	\item $R_3>0,R_4>0$;
	\item $R_3<0,R_4>0$ and $A_1>0, A_2 < 0, A_3>0$,
\end{enumerate}
where $R_3$, $R_4$, $A_1$, $A_2$, and $A_3$ can be found in Appendix.
\end{theorem}

\begin{remark}
From Table \ref{tab:sample-13}, we see that the equilibrium is stable if $R_3>0$ and $R_4>0$. Hence, $R_3>0$, $R_4>0$ is a sufficient condition for the local stability. However, this condition is not necessary. For example, at the first sample point $(1,1/4,1/512)$ listed in Table \ref{tab:sample-13}, the equilibrium is locally stable, but one can verify that $R_3<0$ and $R_4>0$ at this point. Thus, the second condition of Theorem \ref{thm:stable-13-k1k2} is needed.

The necessity of the second condition can also be illustrated by Figure \ref{fig:k1k2-fixedc1} (b, d, f), where the regions defined by the first and second conditions are marked in light grey and dark grey, respectively. By economic intuition, we know that for a fixed value of the marginal cost $c_2$, a decrease in the adjustment speed $k$ would be beneficial to the local stability of the equilibrium. That is to say, the dark grey regions defined by the second condition would be more likely to be included in the stability regions.

It is noted that $A_1$, $A_2$, and $A_3$ involved in the second condition are contained in the so-called generalized discriminant list and can be picked out by repeated trials. Concerning the generalized discriminant list, the reader may refer to \cite{Yang2001A} for more details. The polynomials $A_1$, $A_2$, and $A_3$ are needed here since the condition that $R_3<0$ and $R_4>0$ is not a sufficient condition for the local stability. For example, the model is stable at $(1,1/4,1/512)$, where $R_3<0$ and $R_4>0$. But, the model is unstable at $(1,1/4,34)$, where $R_3<0$ and $R_4>0$ are also satisfied. Consequently, additional polynomials are needed to constrict the region defined by $R_3<0$ and $R_4>0$ such that the complete stability conditions can be acquired.
\end{remark}

\begin{table}[htbp]
	\centering 
	\caption{Selected Sample Points in $\{(c_1,c_2,k)\,|\,c_1> 0, c_2>0, k>0\}$ for $\alpha=1/3$ }
	\label{tab:sample-13} 
	\begin{tabular}{|l|c|c|c||l|c|c|c|}  
		\hline  
		$(c_1,c_2,k)$ & stable & $R_3$ & $R_4$ & $(c_1,c_2,k)$ &  stable & $R_3$ & $R_4$ \\ \hline
$(1, 1/4, 1/512)$& yes & $-$ & $+$ & $(1, 3/8, 1/128)$& yes & $-$ & $+$ \\ \hline
$(1, 1/4, 1)$& yes & $+$ & $+$ & $(1, 3/8, 1)$& yes & $+$ & $+$ \\ \hline
$(1, 1/4, 34)$& no & $-$ & $+$ & $(1, 3/8, 64)$& no & $-$ & $+$ \\ \hline
$(1, 1/4, 153)$& no & $-$ & $-$ & $(1, 3/8, 175)$& no & $-$ & $-$ \\ \hline
$(1, 1/4, 273)$& no & $+$ & $-$ & $(1, 3/8, 287)$& no & $+$ & $-$ \\ \hline
 
$(1, 5/8, 1/32)$& yes & $-$ & $+$ & $(1, 7/8, 1/128)$& yes & $-$ & $+$ \\ \hline
$(1, 5/8, 1)$& yes & $+$ & $+$ & $(1, 7/8, 1)$& yes & $+$ & $+$ \\ \hline
$(1, 5/8, 145)$& no & $-$ & $+$ & $(1, 7/8, 244)$& no & $-$ & $+$ \\ \hline
$(1, 5/8, 231)$& no & $-$ & $-$ & $(1, 7/8, 302)$& no & $-$ & $-$ \\ \hline
$(1, 5/8, 317)$& no & $+$ & $-$ & $(1, 7/8, 361)$& no & $+$ & $-$ \\ \hline

$(1, 5/4, 1/32)$& yes & $-$ & $+$ & $(1, 3/2, 1/16)$& yes & $-$ & $+$ \\ \hline
$(1, 5/4, 1)$& yes & $+$ & $+$ & $(1, 3/2, 1)$& yes & $+$ & $+$ \\ \hline
$(1, 5/4, 335)$& no & $-$ & $+$ & $(1, 3/2, 362)$& no & $-$ & $+$ \\ \hline
$(1, 5/4, 436)$& no & $-$ & $-$ & $(1, 3/2, 544)$& no & $-$ & $-$ \\ \hline
$(1, 5/4, 538)$& no & $+$ & $-$ & $(1, 3/2, 726)$& no & $+$ & $-$ \\ \hline

$(1, 2, 1/16)$& yes & $-$ & $+$ & $(1, 3, 1/16)$& yes & $-$ & $+$ \\ \hline
$(1, 2, 1)$& yes & $+$ & $+$ & $(1, 3, 1)$& yes & $+$ & $+$ \\ \hline
$(1, 2, 403)$& no & $-$ & $+$ & $(1, 3, 471)$& no & $-$ & $+$ \\ \hline
$(1, 2, 804)$& no & $-$ & $-$ & $(1, 3, 1503)$& no & $-$ & $-$ \\ \hline
$(1, 2, 1205)$& no & $+$ & $-$ & $(1, 3, 2536)$& no & $+$ & $-$ \\ \hline
	\end{tabular}
\end{table}

\section{Influence of the Substitutability Degree}

% In the case of $c_1=c_2=c$, it is known that the unique non-vanishing equilibrium of the model for $\alpha=1/2$ is $(3\,c,3\,c)$. According to \eqref{eq:profit}, the corresponding profits of the two firms would be
% $$\Pi_1=\Pi_2=(3\,c-c)\frac{3\,c}{3\,c\cdot (3\,c+3\,c)}=\frac{1}{3}.$$
% For $\alpha=1/3$, we know the non-vanishing equilibrium is $(5\,c,5\,c)$ and the profits of the firms become
% $$\Pi_1=\Pi_2=(5\,c-c)\frac{\sqrt{5\,c}}{5\,c\cdot (\sqrt{5\,c}+\sqrt{5\,c})}=\frac{2}{5}.$$
% Hence, for the Bertrand duopoly model considered in this paper, we conjecture that a decrease in the substitutability degree $\alpha$ may lead to a non-vanishing equilibrium with higher prices and more profits. In other words, product differentiation could increase the prices of the goods and the profits of the involved companies, which is consistent with our economic intuition.

Firstly, we analyze the influence of the substitutability degree $\alpha$ on the size of the stability region of the equilibrium. We start by considering the special case of $c_1=c_2$. 

\begin{proposition}\label{prop:reg-compare-23}
	Let $c_1=c_2$. The stability region for $\alpha=1/2$ is a proper subset of that for $\alpha=1/3$. 
\end{proposition}

\begin{proof}
Recall Theorems \ref{thm:stability12-c} and \ref{thm:stability13-c}. We need to prove that 
$$\frac{2\,k_1+2\,k_2+\sqrt{4\, k_1^{2}-7\, k_1 k_2+4\, k_2^{2}}}{216}>\frac{3\,k_1+3\,k_2+\sqrt{9\, k_1^{2}-17\, k_1 k_2+9\, k_2^{2}}}{2000},$$
which is equivalent to 
$$\left(\frac{2\,k_1+2\,k_2+\sqrt{4\, k_1^{2}-7\, k_1 k_2+4\, k_2^{2}}}{216}\right)^2
-\left(\frac{3\,k_1+3\,k_2+\sqrt{9\, k_1^{2}-17\, k_1 k_2+9\, k_2^{2}}}{2000}\right)^2>0.$$

The left-hand side of the above inequality can be simplified into
\begin{align*}
&-\frac{\left(4374\, k_1+4374\, k_2\right) \sqrt{9\, k_1^{2}-17\, k_1 k_2+9\, k_2^{2}}}{2916000000}
+\frac{\left(250000\, k_1+250000\, k_2\right) \sqrt{4\, k_1^{2}-7\, k_1 k_2+4\, k_2^{2}}}{2916000000}\\
&+\frac{243439\, k_1^{2}}{1458000000}+\frac{61771\, k_1 k_2}{2916000000}+\frac{243439\, k_2^{2}}{1458000000}.
\end{align*}
It is easy to check that
$$\frac{\left(4374\, k_1+4374\, k_2\right) \sqrt{9\, k_1^{2}-17\, k_1 k_2+9\, k_2^{2}}}{2916000000}
<\frac{\left(250000\, k_1+250000\, k_2\right) \sqrt{4\, k_1^{2}-7\, k_1 k_2+4\, k_2^{2}}}{2916000000},$$
which completes the proof.

\end{proof}

If $c_1\neq c_2$, however, the conclusion of the above proposition would be incorrect. For example, if we assume $k_1=k_2=k$ and take $(c_1,c_2,k)=(261/65536, 1/2, 79/1024)$, then 
\begin{align*}
&R_1=\frac{588713082686404258452596575293972215811486125608829}{6129982163463555433433388108601236734474956488734408704}>0,\\
&R_2=
\frac{108130364702270905134254005155560019343}{340282366920938463463374607431768211456}>0.
\end{align*}
Hence, $(261/65536, 1/2, 79/1024)$ is in the stability region of the model for $\alpha=1/2$. But, at the same parameter point, namely $(c_1,c_2,k)=(261/65536, 1/2, 79/1024)$, we have
\begin{align*}
&R_3=-\frac{791461358900213183480020700044263844445257635142615074110540187}{26328072917139296674479506920917608079723773850137277813577744384}<0,\\
&R_4=\frac{526438846625624761986017962528229497389068363385599391}{374144419156711147060143317175368453031918731001856}>0,	
\end{align*}
and
\begin{align*}
&A_1=\frac{44864955}{4294967296}>0,~~A_2=-\frac{842240947483983714275440267}{81129638414606681695789005144064}<0,\\
&A_3=-\frac{63936547182666560163845458457577}{649037107316853453566312041152512}	<0.
\end{align*}
This means that the stability conditions of Theorem \ref{thm:stable-13-k1k2} for $\alpha=1/3$ are not satisfied.

On the other hand, one can also find some points where the model is stable for $\alpha=1/3$ but unstable for $\alpha=1/2$. For example, at $(c_1,c_2,k)=( 3/8,  1/2,  827/64)$, we know
\begin{align*}
	R_3=\frac{40079185741889580295152003015}{288230376151711744}>0,~~R_4=\frac{29339436396656781}{17179869184}>0.
\end{align*}
Therefore, $( 3/8,  1/2,  827/64)$ is in the stability region for $\alpha=1/3$.
However,
\begin{align*}
	R_1=-\frac{24200272602071108539}{17592186044416}<0,~~R_2=-\frac{96467864887}{67108864}<0.
\end{align*}
That is to say, $( 3/8,  1/2,  827/64)$ is an unstable parameter point for $\alpha=1/2$. 

Figure \ref{fig:k1k2-fixedk} depicts the 2-dimensional cross-sections of the stability regions for $\alpha=1/2$ and $\alpha=1/3$. For comparison purposes, we place the cross-sections for $\alpha=1/2$ on the left and those for $\alpha=1/3$ on the right. We set $k_1=k_2=k$ and choose three different values of the parameter $k$, i.e., $k=1/2,1,10$, to observe the effect of variation of $k$ on the size of the stability regions. The curves of $R_1=0$ and $R_3=0$ are marked in blue; the curves of $R_2=0$ and $R_3$ are marked in green; the curves of $A_1=0$, $A_2=0$ and $A_3=0$ are marked in red. The stability regions are colored in light grey. From Figure \ref{fig:k1k2-fixedk}, we find that the stability region would shrink if the firms react or adjust their outputs faster both for $\alpha=1/2$ and $\alpha=1/3$. Similarly, in Figure \ref{fig:k1k2-fixedc1}, we assume that  $k_1$ and $k_2$ are identical and choose three different values of $c_1$, i.e., $c_1=1/2,1,10$. The regions of $R_1>0$, $R_2>0$ and those of $R_3>0$, $R_4>0$ are colored in light grey, while the regions defined by $R_3<0$, $R_4>0$, $A_1>0$, $A_2<0$, $A_3>0$ are colored in dark grey. From Figure \ref{fig:k1k2-fixedc1}, we observe that increasing the marginal cost $c_1$ of the first firm could result in the enlargement of the stability region for $\alpha=1/2$ and $\alpha=1/3$.

As aforementioned, in the case of $c_1\neq c_2$ and $k_1=k_2$, it can not be proved that the stability region for $\alpha=1/3$ covers that for $\alpha=1/2$. From Figures \ref{fig:k1k2-fixedk} and \ref{fig:k1k2-fixedc1}, however, it seems that the stability region for $\alpha=1/3$ is larger than that for $\alpha=1/2$. Consequently, for the Bertrand duopoly model considered in this paper, we may conclude that increasing the substitutability degree $\alpha$ has an effect of destabilizing the unique non-vanishing equilibrium in some sense. In other words, product differentiation might make the considered model more stable, which is an important finding from an economic point of view. 
{Shy \cite{Shy1995I} discussed the traditional view on the degree of product differentiation, i.e., a decrease in product differentiation may result in an increase in market competition intensity and even a price war among involved firms. The possible explanation for our finding is that a price war might destabilize the equilibrium of the Bertrand game with differentiated goods.}
It should be noted that our conclusion is in contrast with the one by Agliari et al.\  \cite{Agliari2016N}. Specifically, Agliari et al.\ \cite{Agliari2016N} investigated a Cournot duopoly model with differentiated products and employed the same CES utility function and the same linear cost functions as in our study. However, they discovered that a higher degree of product differentiation or a lower degree of substitutability leads to the destabilization of their model. This contradiction may help reveal the essential difference between the Bertrand and Cournot oligopolies with differentiated goods.

\begin{figure}[htbp]
  \centering
  \subfloat[$\alpha=1/2$, $k=1/2$]{\includegraphics[width=0.35\textwidth]{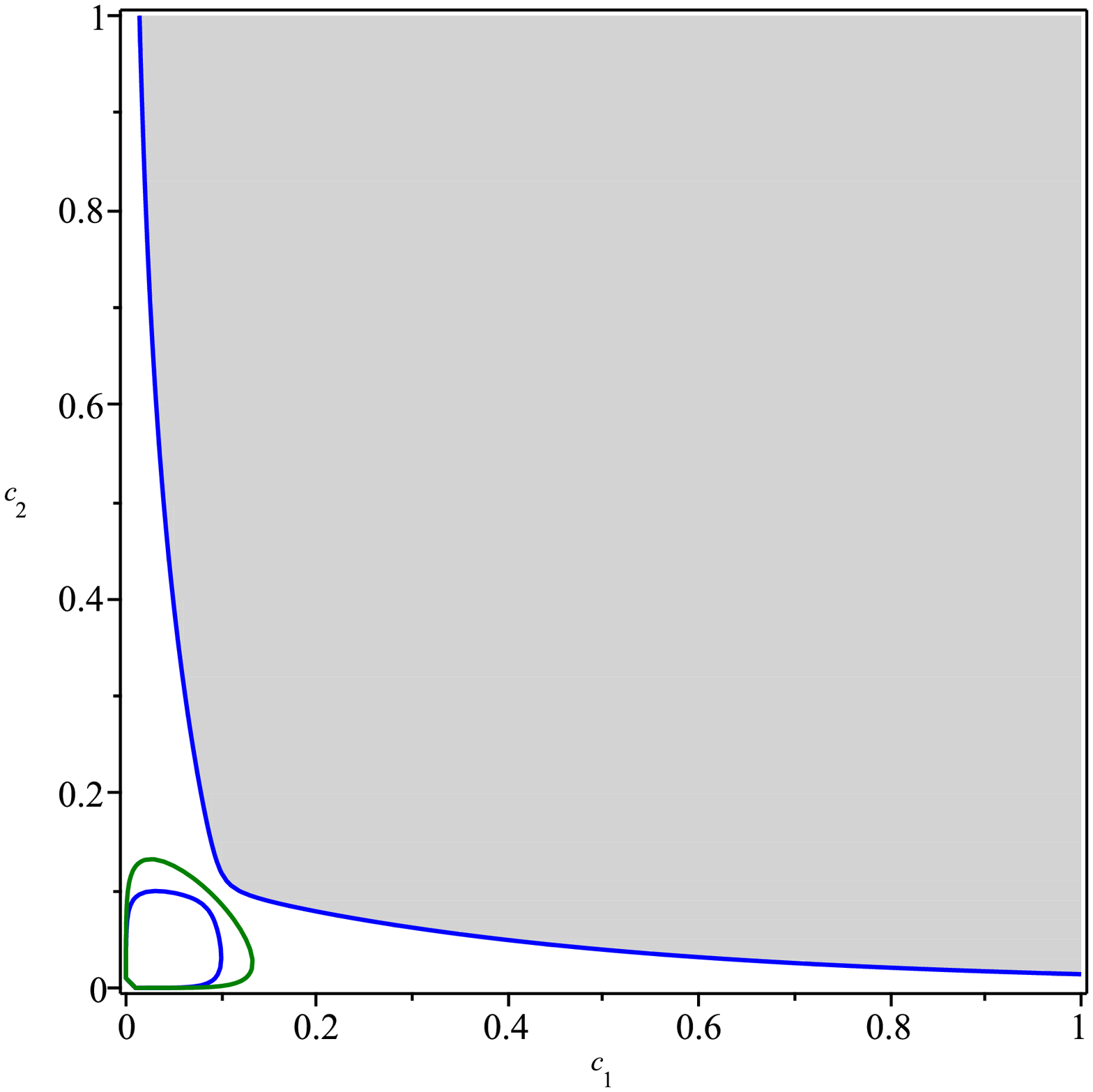}} 
  \subfloat[$\alpha=1/3$, $k=1/2$]{\includegraphics[width=0.35\textwidth]{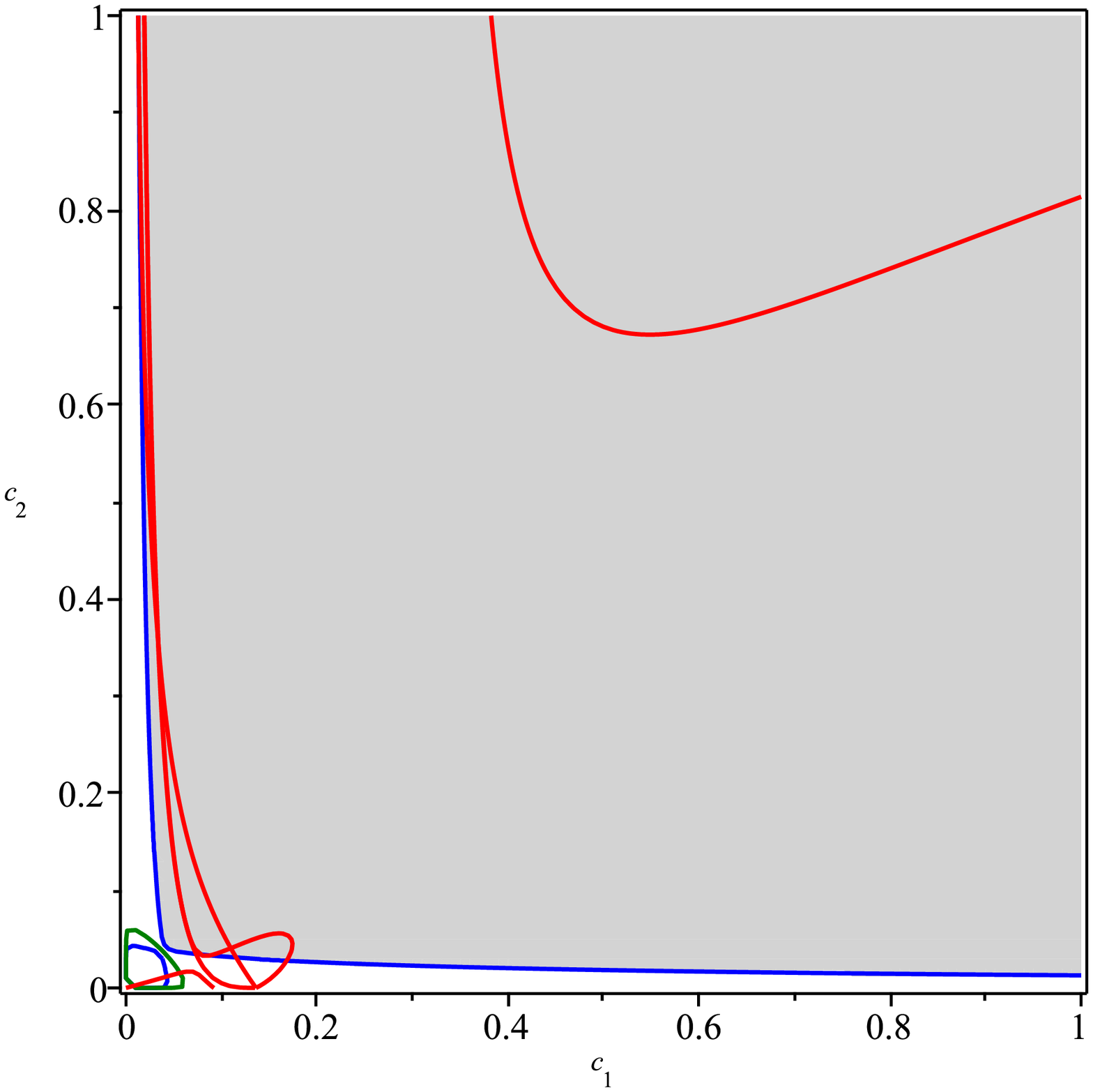}} \\

  \subfloat[$\alpha=1/2$, $k=1$]{\includegraphics[width=0.35\textwidth]{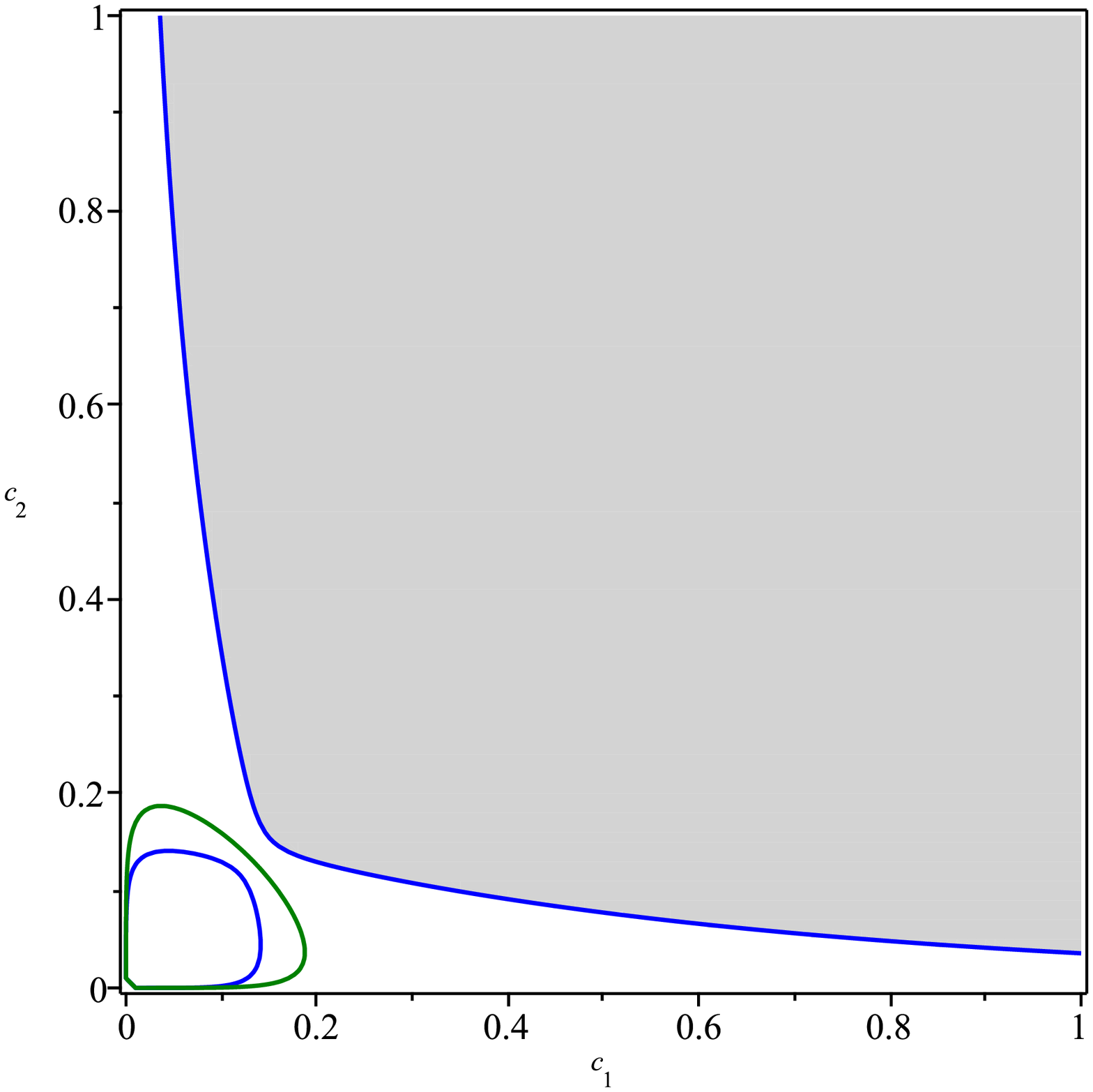}} 
  \subfloat[$\alpha=1/3$, $k=1$]{\includegraphics[width=0.35\textwidth]{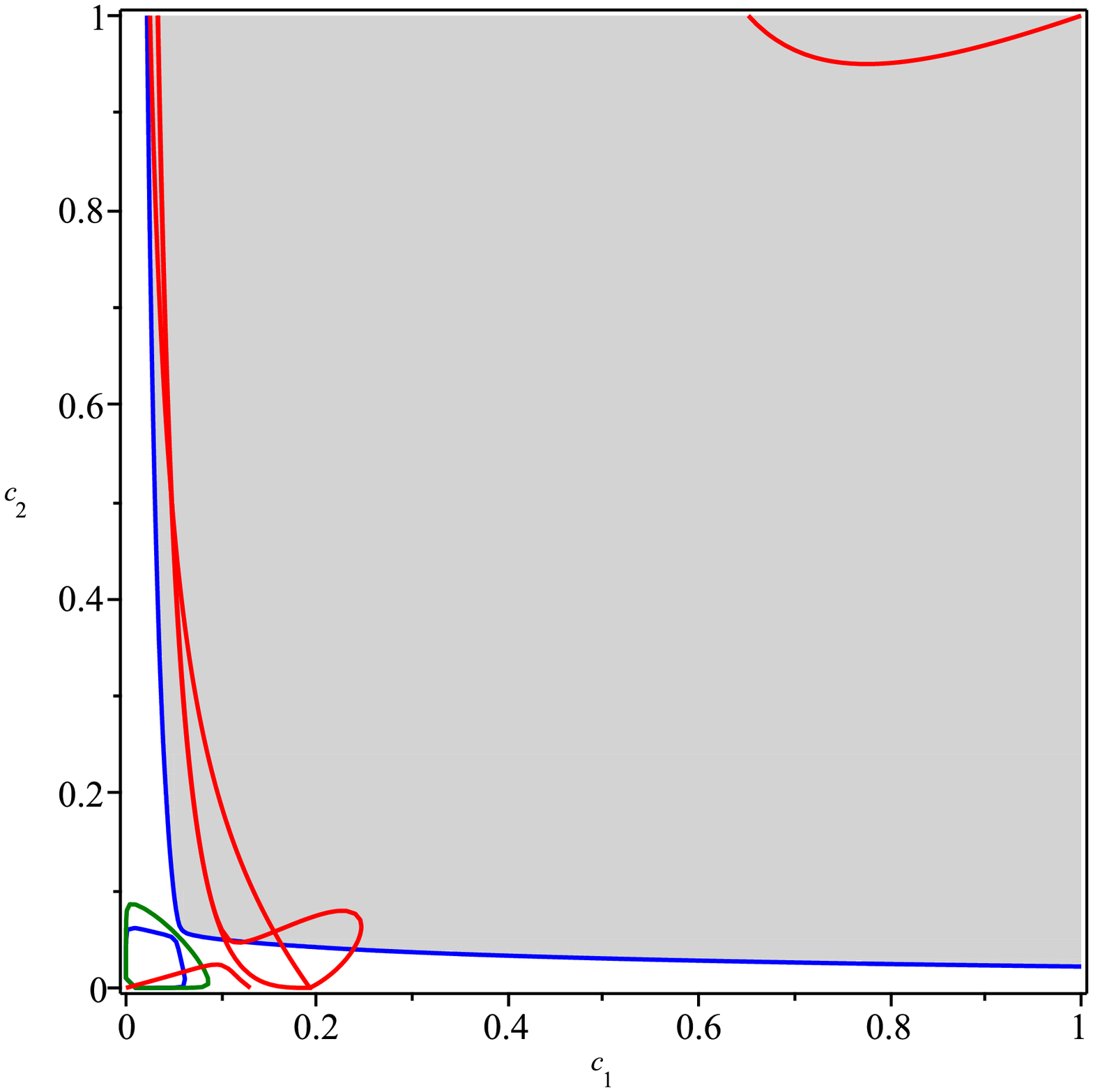}} \\

  \subfloat[$\alpha=1/2$, $k=10$]{\includegraphics[width=0.35\textwidth]{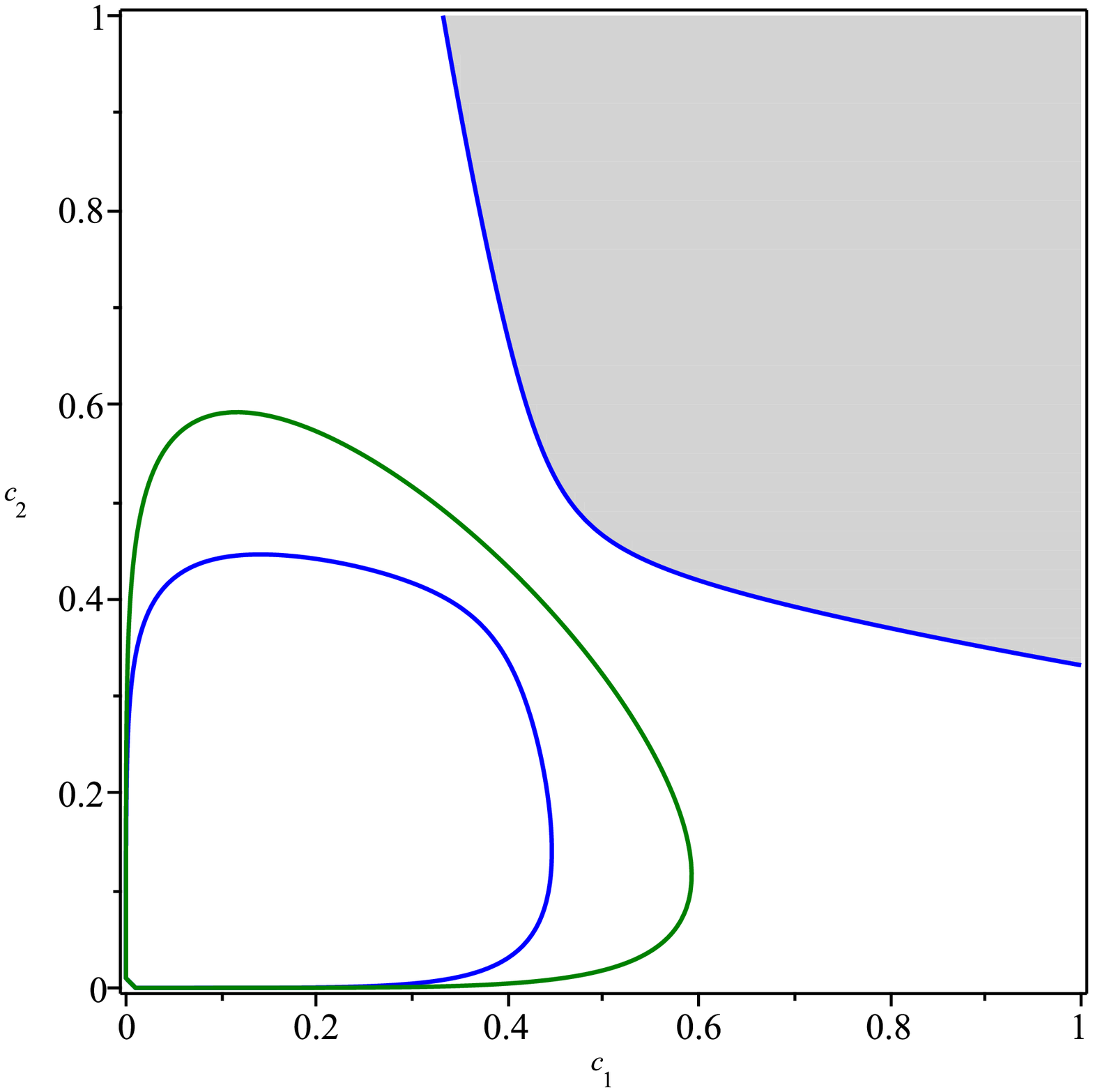}} 
  \subfloat[$\alpha=1/3$, $k=10$]{\includegraphics[width=0.35\textwidth]{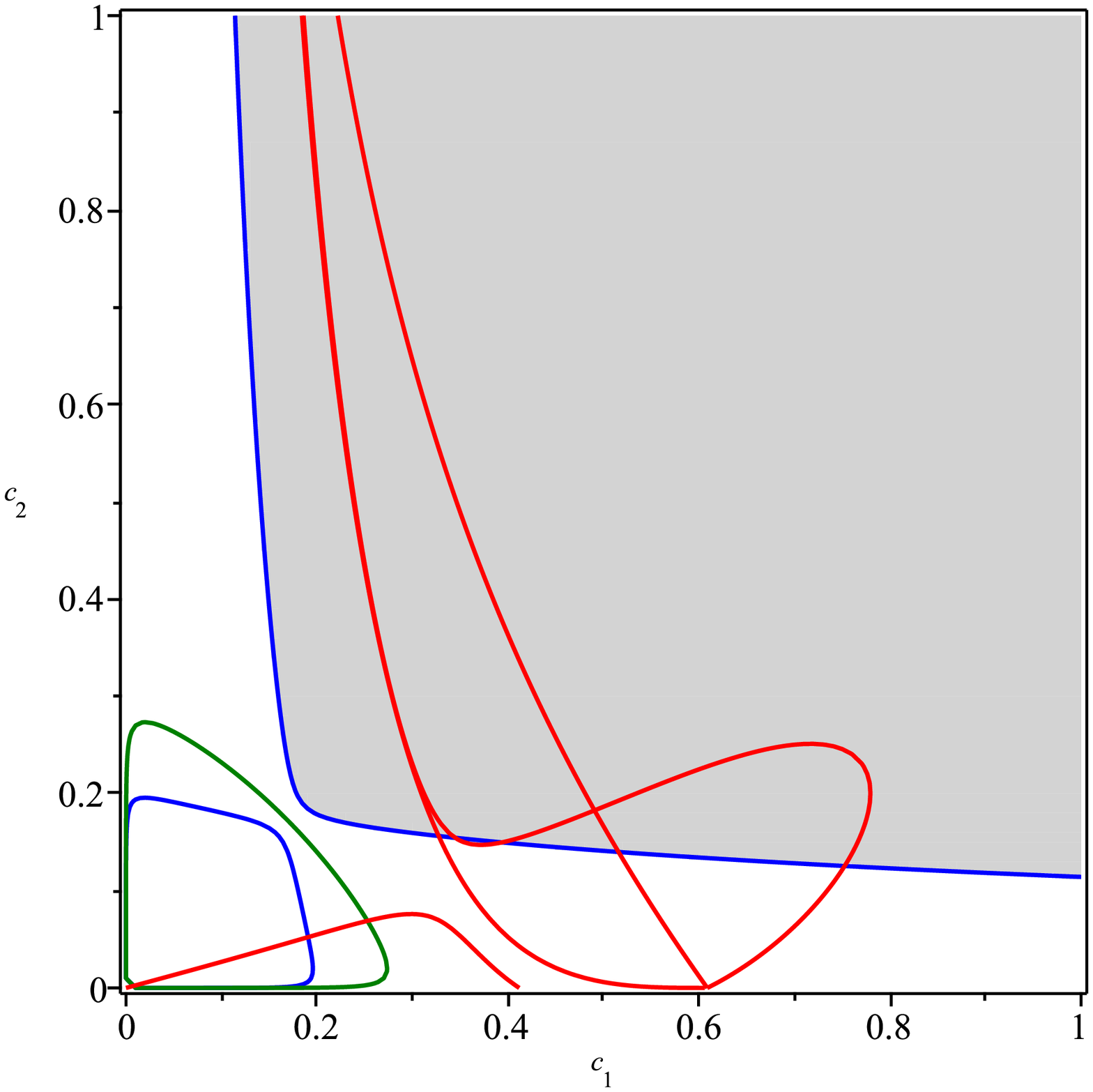}} \\
  \caption{The 2-dimensional cross-sections of the stability regions for $\alpha=1/2$ and $\alpha=1/3$ if we set $k_1=k_2=k$ and fix $k=1/2,1,10$. The curves of $R_1=0$ and $R_3=0$ are marked in blue; the curves of $R_2=0$ and $R_3$ are marked in green; the curves of $A_1=0$, $A_2=0$ and $A_3=0$ are marked in red. The stability regions are colored in light grey.}
 \label{fig:k1k2-fixedk}
\end{figure}

\begin{figure}[htbp]
  \centering
  \subfloat[$\alpha=1/2$, $c_1=1/2$]{\includegraphics[width=0.35\textwidth]{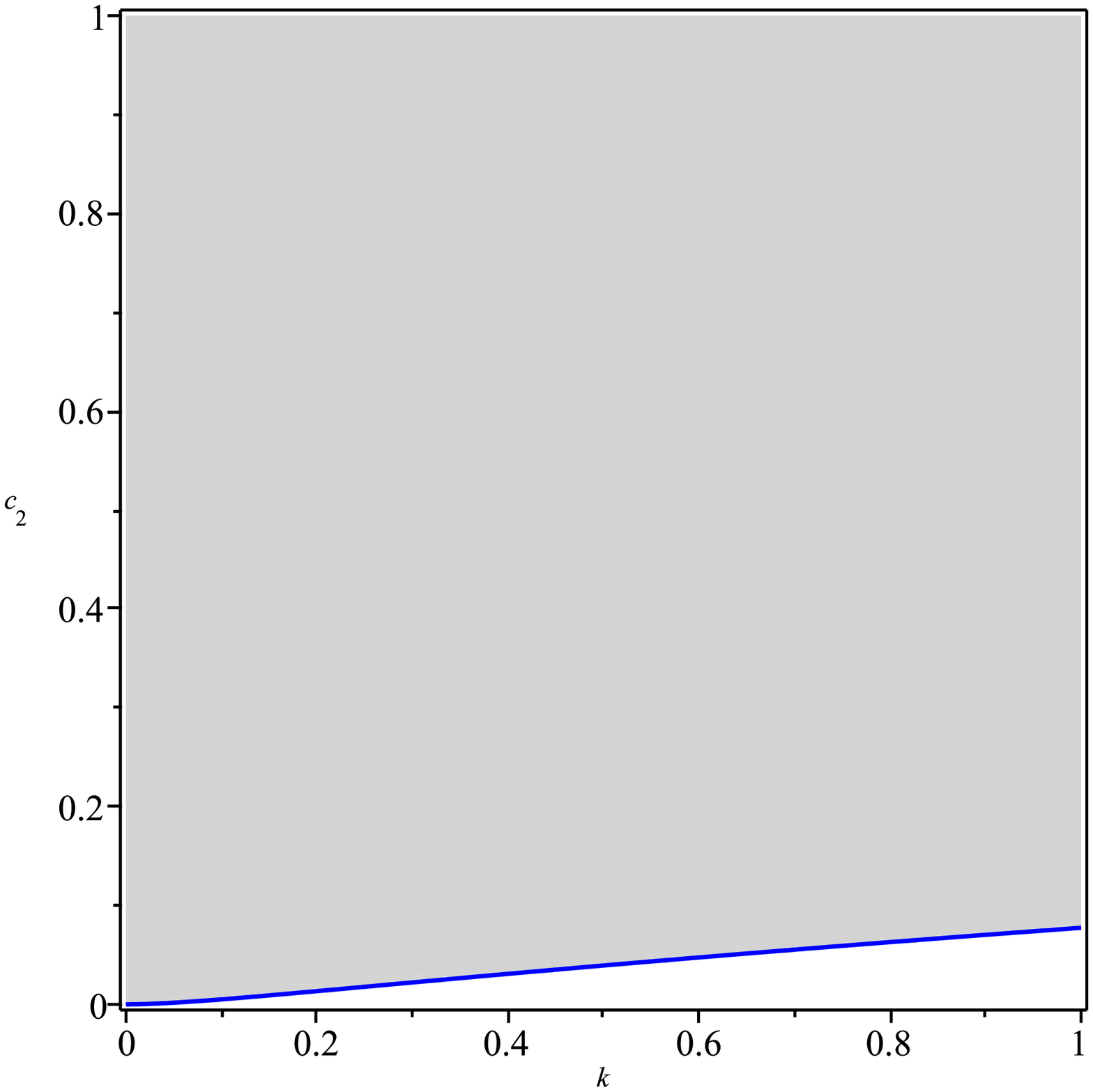}} 
  \subfloat[$\alpha=1/3$, $c_1=1/2$]{\includegraphics[width=0.35\textwidth]{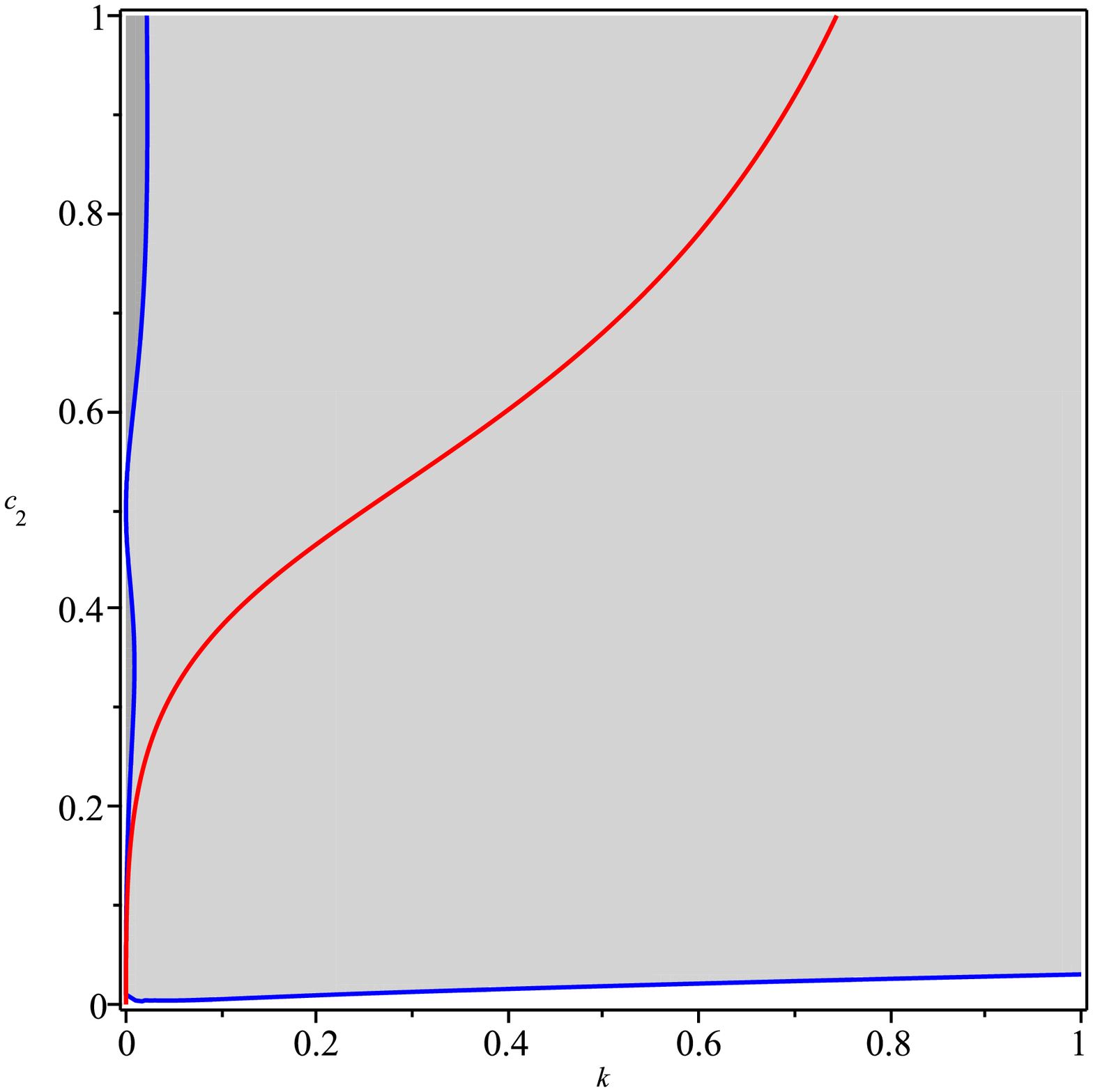}} \\

  \subfloat[$\alpha=1/2$, $c_1=1$]{\includegraphics[width=0.35\textwidth]{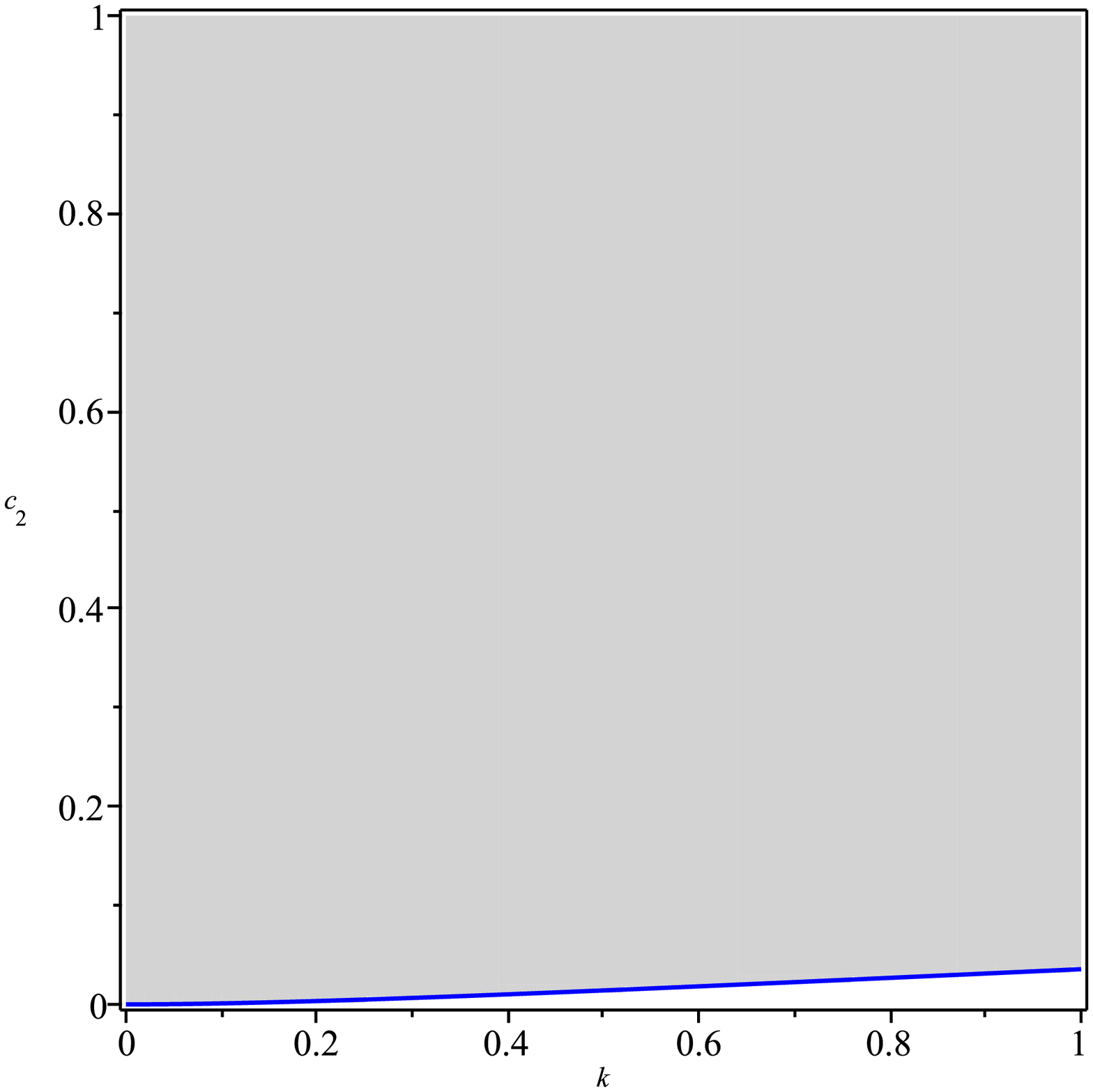}} 
  \subfloat[$\alpha=1/3$, $c_1=1$]{\includegraphics[width=0.35\textwidth]{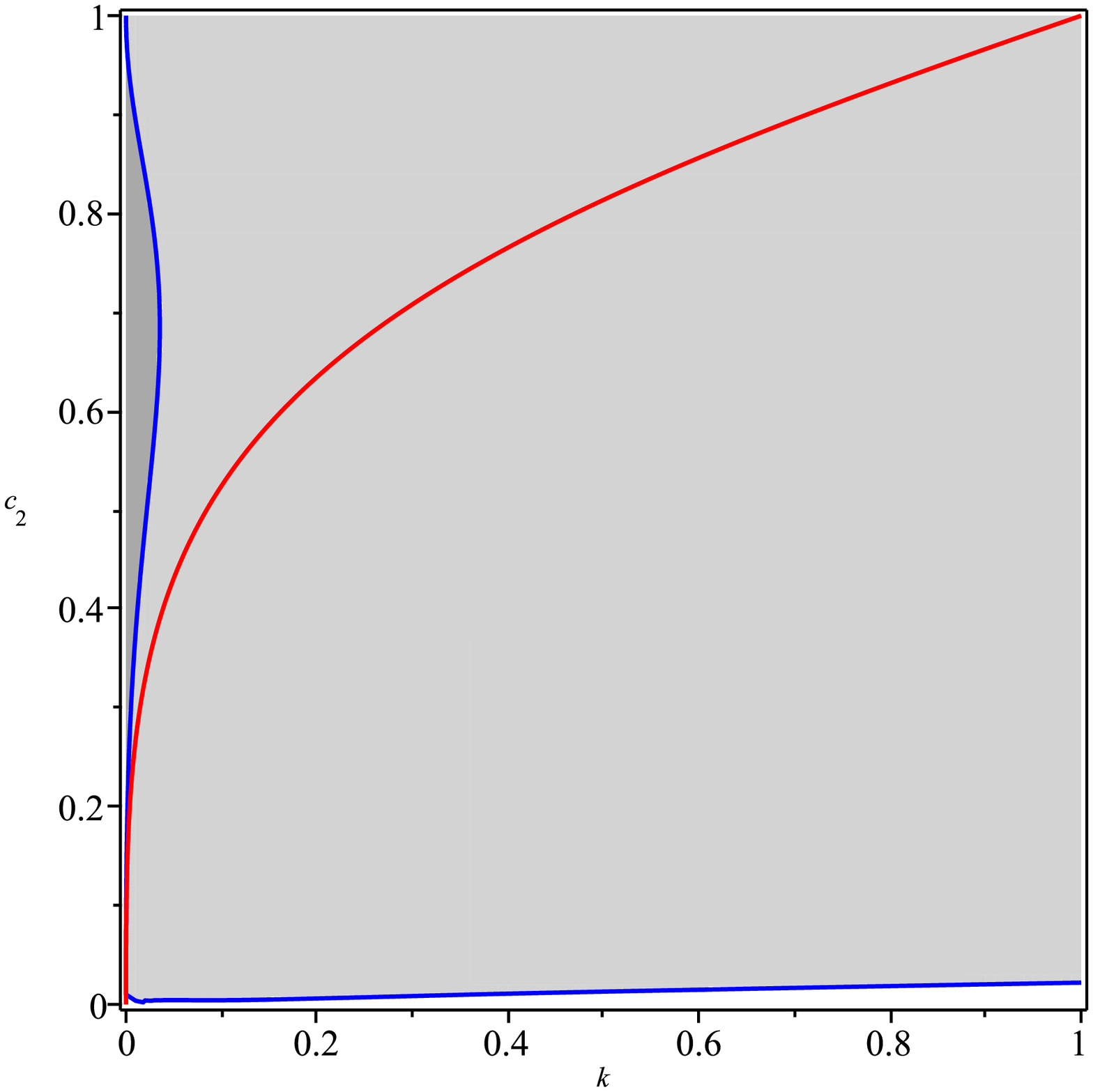}} \\

  \subfloat[$\alpha=1/2$, $c_1=10$]{\includegraphics[width=0.35\textwidth]{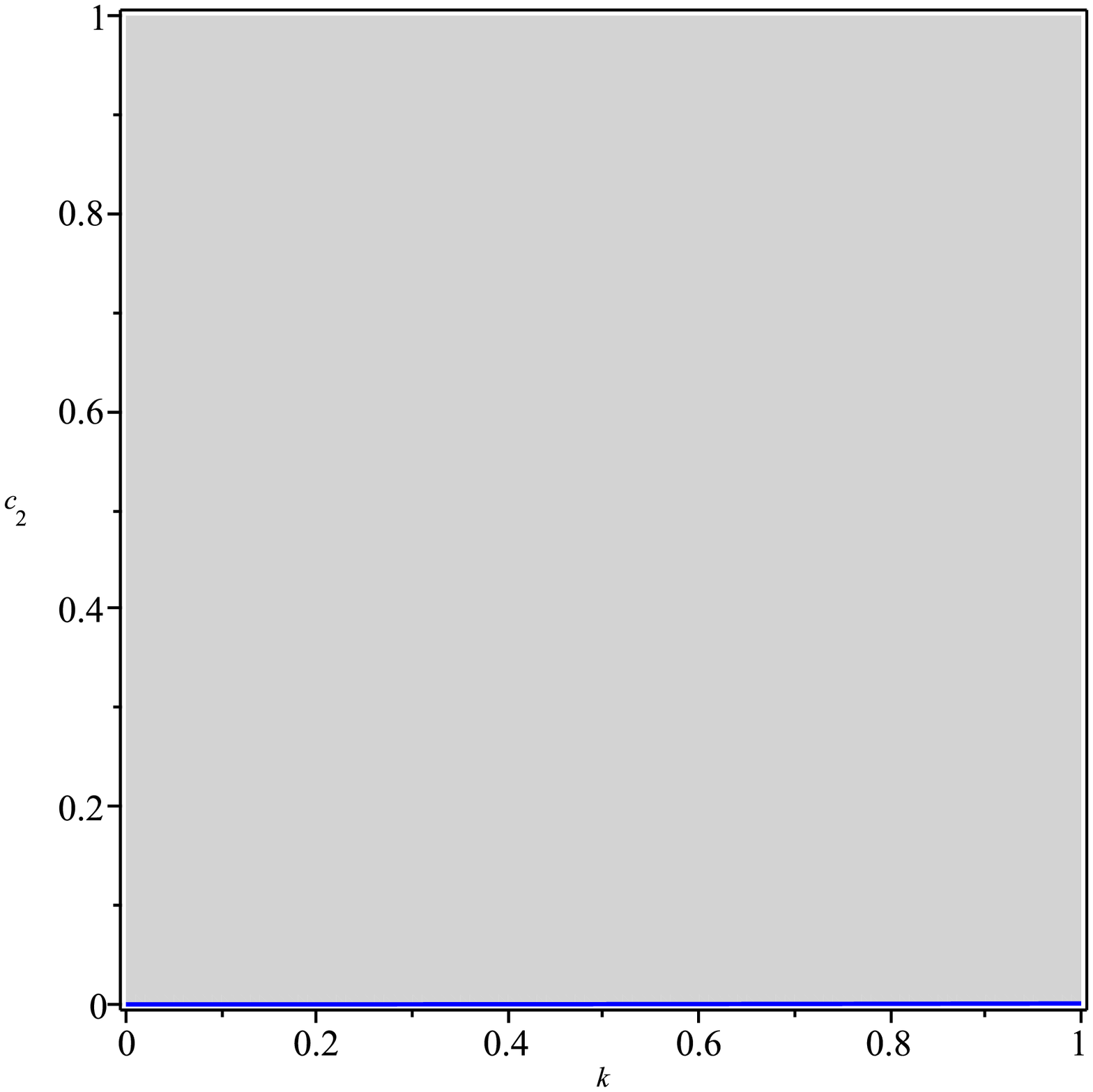}} 
  \subfloat[$\alpha=1/3$, $c_1=10$]{\includegraphics[width=0.35\textwidth]{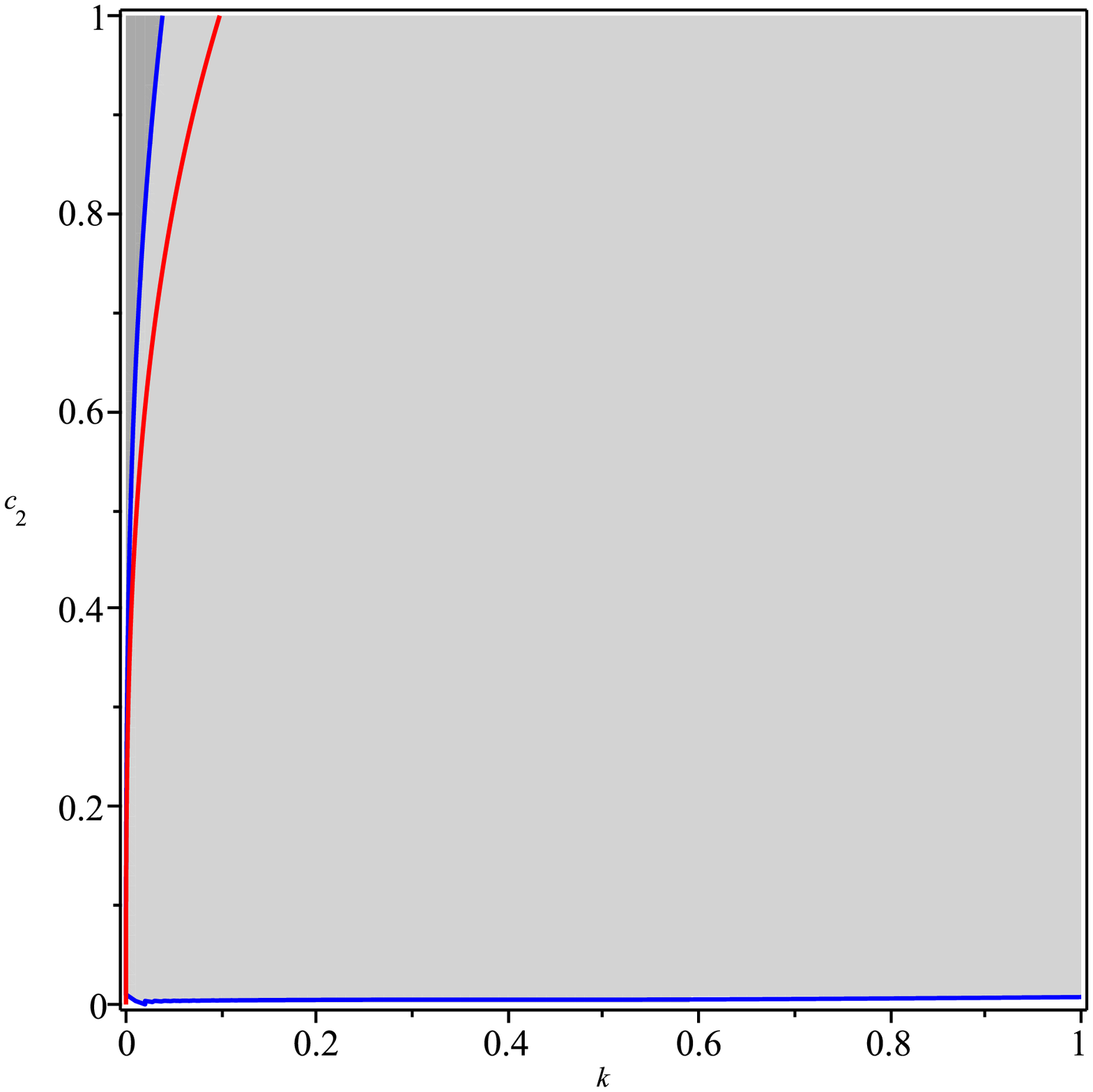}} \\
  \caption{The 2-dimensional cross-sections of the stability regions for $\alpha=1/2$ and $\alpha=1/3$ if we set $k_1=k_2=k$ and fix $c_1=1/2,1,10$. The curves of $R_1=0$ and $R_3=0$ are marked in blue; the curves of $R_2=0$ and $R_3$ are marked in green; the curves of $A_1=0$, $A_2=0$ and $A_3=0$ are marked in red. The regions of $R_1>0$, $R_2>0$ and those of $R_3>0$, $R_4>0$ are colored in light grey, while the regions defined by $R_3<0$, $R_4>0$, $A_1>0$, $A_2<0$, $A_3>0$ are colored in dark grey.}
 \label{fig:k1k2-fixedc1}
\end{figure}

From an economic point of view, the effects on economic variables such as prices and profits of changing the substitutability degree are interesting. In the sequel, we focus on the comparative statics in the special case of identical marginal costs. Let $c_1=c_2=c$. According to \eqref{eq:iter-map}, the equilibrium satisfies that
\begin{equation}\label{eq:general-a}
\left\{
\begin{split}
&-p_2^\beta p_1^{1+\beta} \beta + p_2^{2\beta}c+ (p_1^\beta p_2^\beta	)(1+\beta)c=0,\\
&-p_1^\beta p_2^{1+\beta} \beta + p_1^{2\beta}c+ (p_1^\beta p_2^\beta	)(1+\beta)c=0.
\end{split}
\right.
\end{equation}
Hence, 
$$-p_2^\beta p_1^{1+\beta} \beta + p_2^{2\beta}c=-p_1^\beta p_2^{1+\beta} \beta + p_1^{2\beta}c,$$
which implies that
$$(p_1^{2\beta}-p_2^{2 \beta})c=(p_2-p_1)p_1^\beta p_2^\beta \beta.$$
Without loss of generality, we suppose that $p_1\geq p_2$.
Since $c>0$ and $\beta>0$, we know $(p_1^{2\beta}-p_2^{2 \beta})c\geq 0$ and $(p_2-p_1)p_1^\beta p_2^\beta \beta\leq 0$, which implies $p_1=p_2$. Plugging $p_1=p_2$ into the first equation of \eqref{eq:general-a}, one can solve $p_1=p_2=\frac{c(2+\beta)}{\beta}$. Therefore, at the equilibrium $q_1=q_2=\frac{\beta}{2\, c(2+\beta)}$. As $\beta=\alpha/(1-\alpha)$, we obtain
$$\frac{\partial p_i}{\partial \alpha}= -\frac{2\, c}{\alpha^{2}}<0,~~
\frac{\partial q_i}{\partial \alpha}=\frac{1}{\left(-2+\alpha \right)^{2} c}>0.$$
According to \eqref{eq:profit}, the profits of the two firms would be
$$\Pi_1=\Pi_2=\left(\frac{c(2+\beta)}{\beta}-c\right)\frac{\beta}{2\,c(2+\beta)}=\frac{1}{2+\beta}=1+\frac{1}{\alpha-2}.$$
Hence, for $i=1,2$,
$$\frac{\partial \Pi_i}{\partial \alpha}=-\frac{1}{\left(\alpha-2 \right)^{2}}<0.$$

Recalling the inverse demands \eqref{eq:inv-dem}, for a point $(q_1^*,q_2^*)$ on the indifference curve, we define the consumer surplus of the first product to be
$$CS_1=\int_{0}^{q_1^*}\frac{q_1^{\alpha -1}}{q_1^\alpha+q_2^{*\alpha}} d q_1	
=\frac{1}{\alpha}\int_{0}^{q_1^*}\frac{d (q_1^\alpha+q_2^{*\alpha})}{q_1^\alpha+q_2^{*\alpha}}
=\frac{1}{\alpha}\ln \left[1+\left(\frac{q_1^*}{q_2^*}\right)^\alpha\right].$$
In the case of $c_1=c_2$, the outputs of the two products are equal at the equilibrium. Therefore, we have that $CS_1=CS_2=\frac{1}{\alpha} \ln 2$. Accordingly, the social welfare is 
$$W=CS_1+CS_2+\Pi_1+\Pi_2=\frac{2}{\alpha}\ln 2 +\frac{2}{\alpha -2}+2.$$ 
Then it is known that 
$$\frac{\partial W}{\partial \alpha}= -\frac{2\ln 2}{\alpha^2}-\frac{2}{(\alpha-2)^\alpha}<0.$$

To summarize, in the special case of identical marginal costs, an increase in the substitutability degree $\alpha$ leads to a stable equilibrium with lower prices, higher supplies, lower profits, and lower welfare. In other words, the degree of product differentiation is positively related to the prices of the goods, the profits of the involved companies, and the social welfare, which is consistent with our economic intuition.

\section{Numerical Simulations}

This section provides numerical simulations to illustrate the complex dynamics of the considered Bertrand duopoly model. The first purpose of our simulations is to confirm the main conclusion of Section 5 that increasing the substitutability degree $\alpha$ could destabilize the unique non-vanishing equilibrium. In Figure \ref{fig:bifur_a}, we depict the 1-dimensional bifurcation diagrams with respect to $\alpha$, where we fix the other parameters $k_{1}=k_{2}=1$, $c_{1}=c_{2}=0.2$ and set the initial point to be $(0.56,1.06)$. The bifurcation diagrams against $p_1$ and $p_2$ are given in Figure \ref{fig:bifur_a} $(a,c)$ and $(b,d)$, respectively. It is observed that complex dynamics appear when $\alpha$ becomes large enough. Specifically, there exists one unique stable equilibrium at first, then a stable 2-cycle orbit, and finally a chaotic set as $\alpha$ varies from $0.1$ up to $0.7$. To show the transition clearly, the 1-dimensional bifurcation diagrams are enlarged for $\alpha\in (0.55,0.6)$ in $(c,d)$. One can see that, when $\alpha=0.553372$, a branching point occurs and the unique fixed point bifurcates into a 2-cycle orbit, which, however, is not a period-doubling bifurcation point. This 2-cycle orbit loses its stability through a Neimark-Sacker bifurcation rather than a period-doubling bifurcation at $\alpha=0.577570$.

More details can be found in Figure \ref{fig:pp-a}, where we plot the phase portraits for $k_1=k_2=1$ and $c_1=c_2=0.2$ with the initial point $(0.56,1.06)$. From Figure \ref{fig:pp-a} (a), we observe that, after the occurrence of a Neimark-Sacker bifurcation, the 2-cycle orbit ($P_{21}(0.464194,0.607384)$ and $P_{21}(0.607384,0.464194)$) becomes unstable and bifurcates into two invariant closed orbits when $\alpha=0.58$; the unique equilibrium $E_{1}(0.492557,0.492557)$ goes to $E_{1new}(0.489655,0.489655)$ when $\alpha=0.58$. Furthermore, all points on the diagonal line $x=y$ converge to $E_{1new}$. The two invariant closed orbits marked in blue are stable and points converge to them from inside and outside. Figure \ref{fig:pp-a} (b) depicts the phase portrait when $\alpha=0.59$ and the other parameters are set to be the same as (a). From (b), one can discover chaotic attractors with symmetry. 
%In short, we explore the influence of the variation of the parameter $\alpha$ on the dynamic behavior of the considered Bertrand model if we keep $k_{1},k_{2},c_{1},c_{2}$ fixed. 
The above observations show that an increase in the substitutability degree $\alpha$ leads to the emergence of instability, complex dynamics, and even chaos in the considered model.

%\begin{figure}[htbp]
%  \centering
%{\includegraphics[width=0.7\textwidth]{fig/bifur_case_1_alpha_p1.eps}}\\
%{\includegraphics[width=0.7\textwidth]{fig/bifur_case_1_alpha_p2.eps}} \\
%   \caption{Here, $k_{1}=k_{2}=1$,$c_{1}=c_{2}=0.2$, $\alpha\in (0.1,0.7)$ and initial point (0.56,1.06).}
% \label{fig:bifur_case_1}
%\end{figure}
%
%\begin{figure}[htbp]
%  \centering
%  {\includegraphics[width=0.7\textwidth]{fig/bifur_case_1_alpha_p1_z.eps}}\\
%  {\includegraphics[width=0.7\textwidth]{fig/bifur_case_1_alpha_p2_z.eps}} \\
%  \caption{Local amplification for $\alpha\in (0.55,0.6)$}
% \label{fig:bifur_case_1_z}
%\end{figure}
%
%

\begin{figure}[htbp]
  \centering
  \subfloat[against $p_1$]{\includegraphics[width=0.45\textwidth]{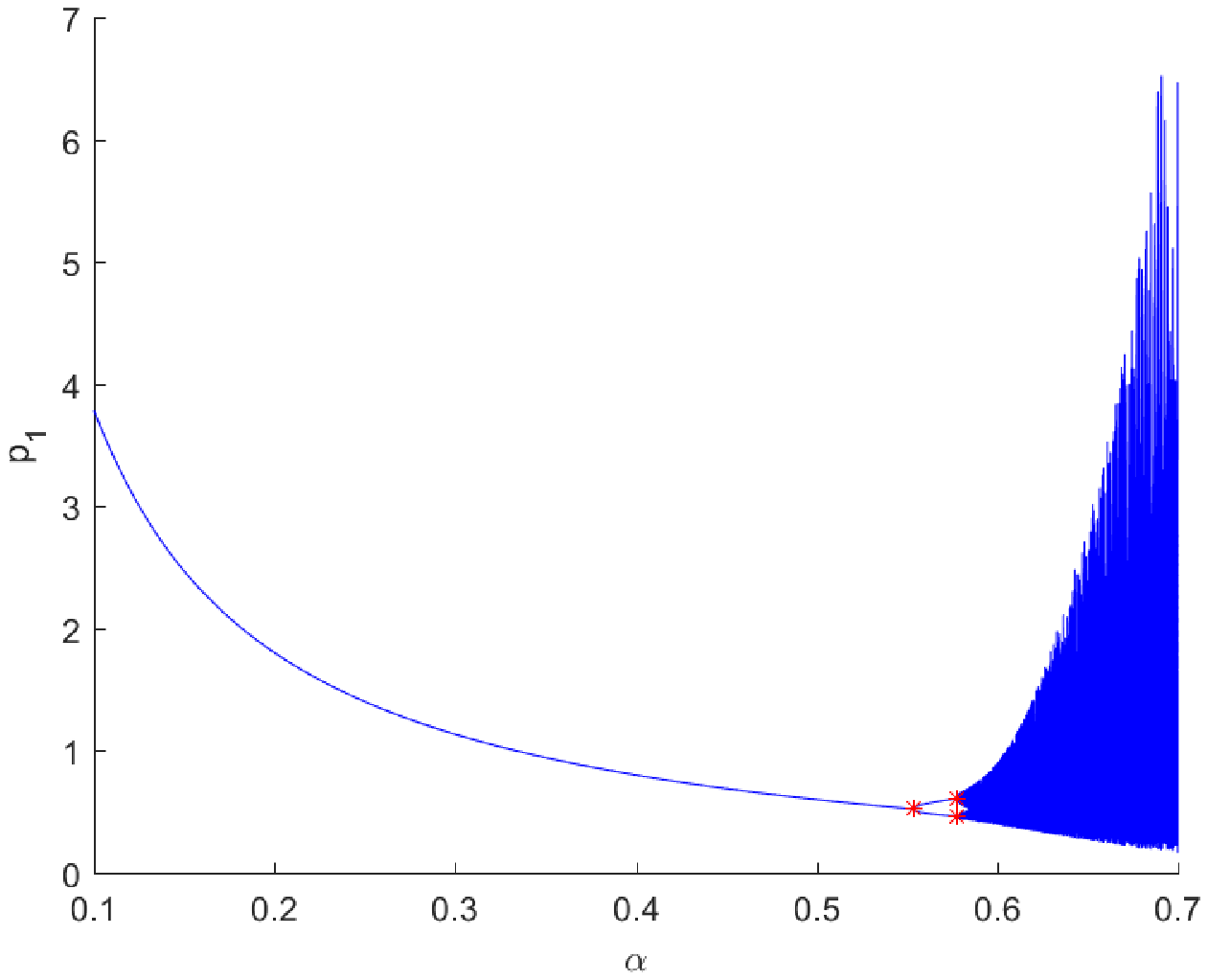}} 
  \subfloat[against $p_2$]{\includegraphics[width=0.45\textwidth]{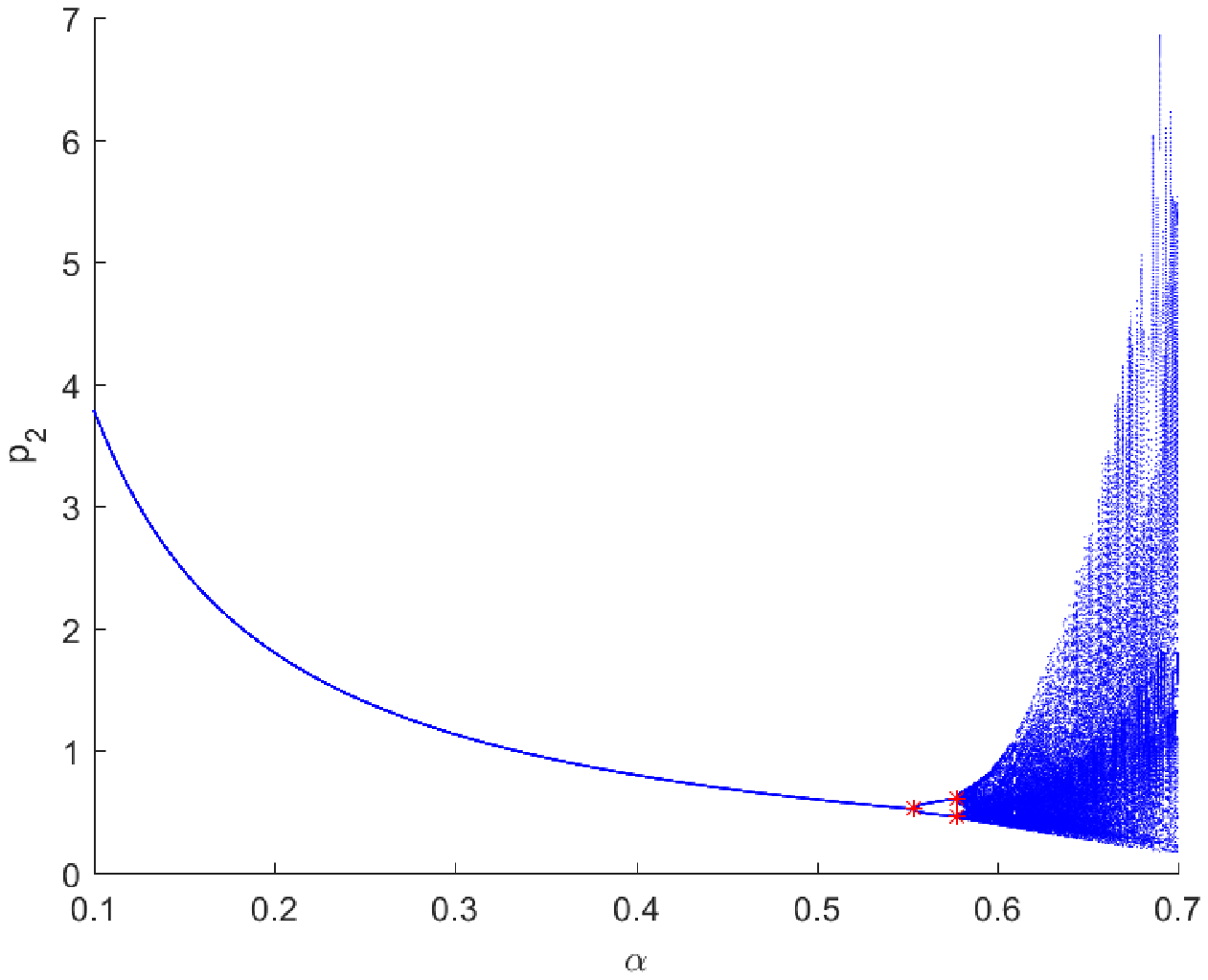}} \\

  \subfloat[against $p_1$ and enlarged for $\alpha\in (0.55,0.6)$]{\includegraphics[width=0.45\textwidth]{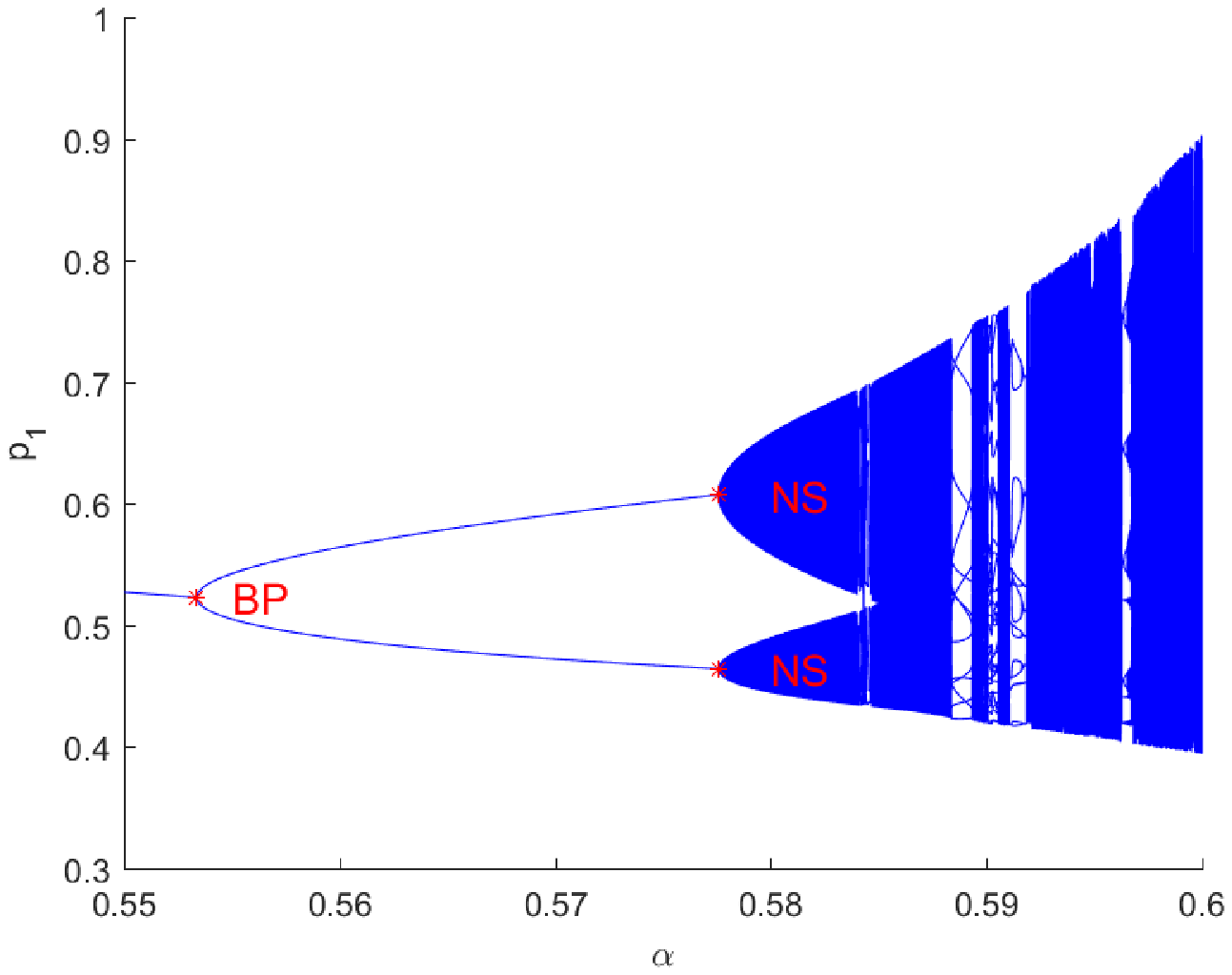}} 
  \subfloat[against $p_2$ and enlarged for $\alpha\in (0.55,0.6)$]{\includegraphics[width=0.45\textwidth]{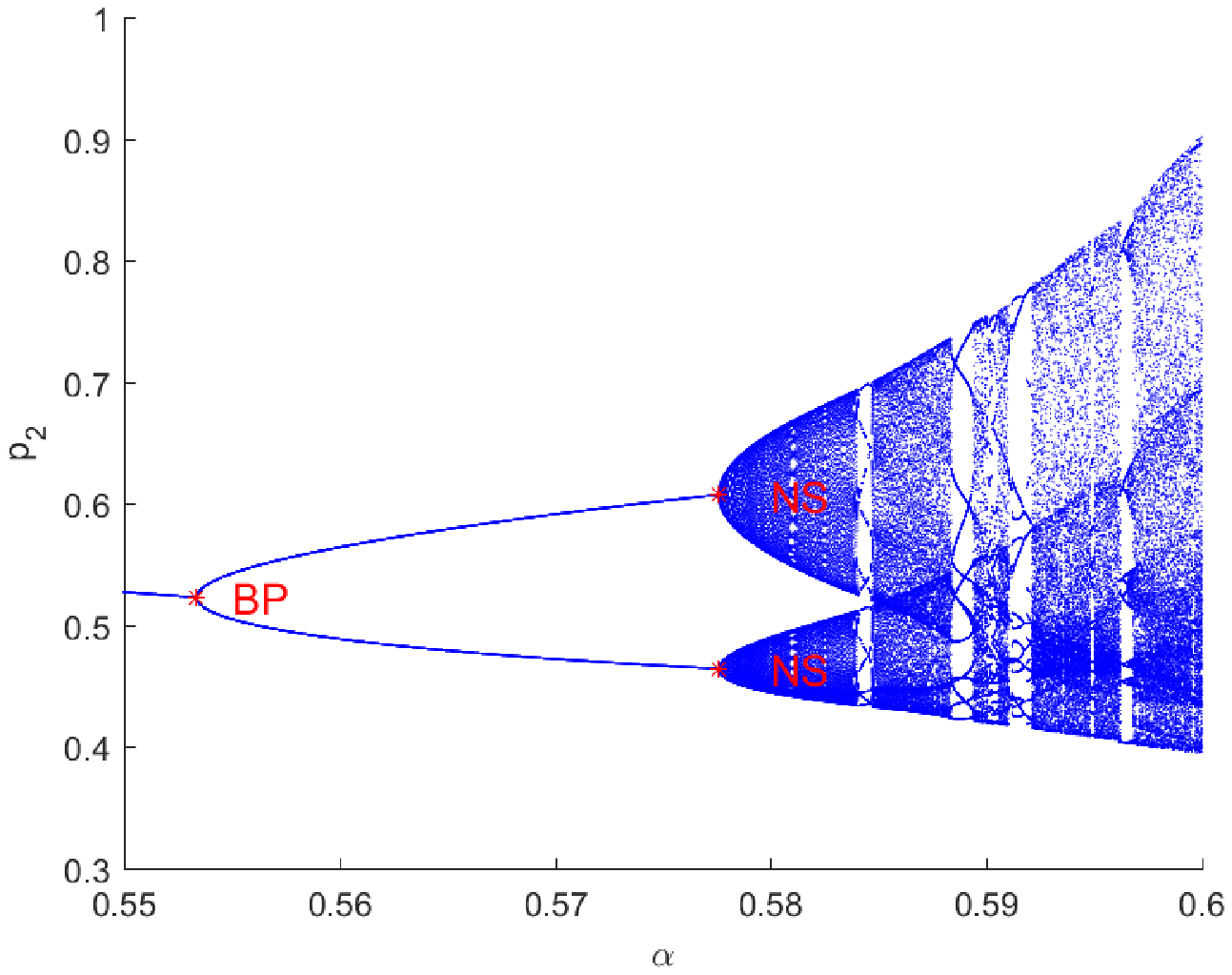}} \\

  \caption{The 1-dimensional bifurcation diagrams with respect to $\alpha$ if we fix $k_{1}=k_{2}=1$, $c_{1}=c_{2}=0.2$ and set the initial point to be $(0.56,1.06)$. }
 \label{fig:bifur_a}
\end{figure}

\begin{figure}[htbp]
  \centering
  \subfloat[$\alpha=0.58$]{\includegraphics[width=0.45\textwidth]{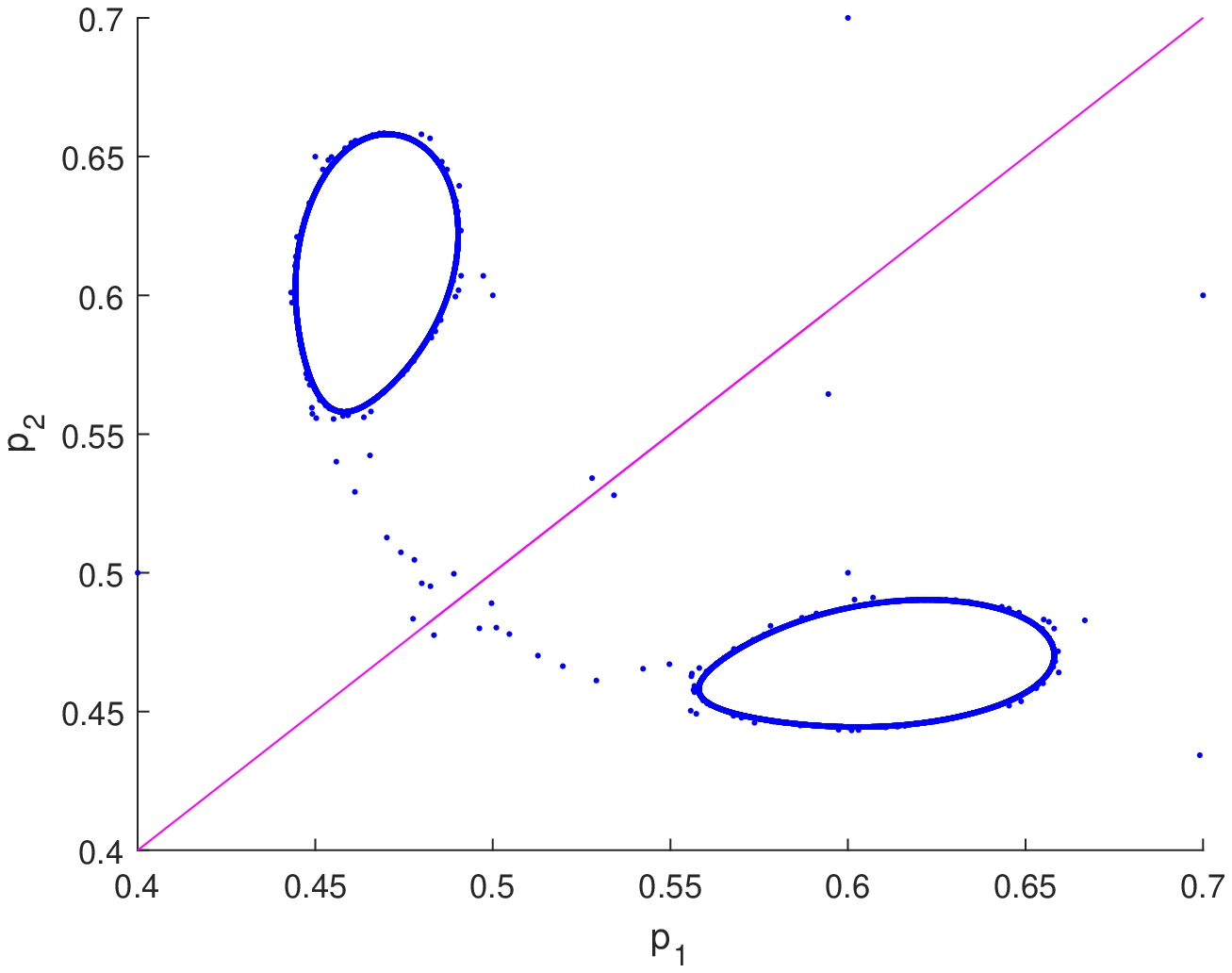}} 
  \subfloat[$\alpha=0.59$]{\includegraphics[width=0.45\textwidth]{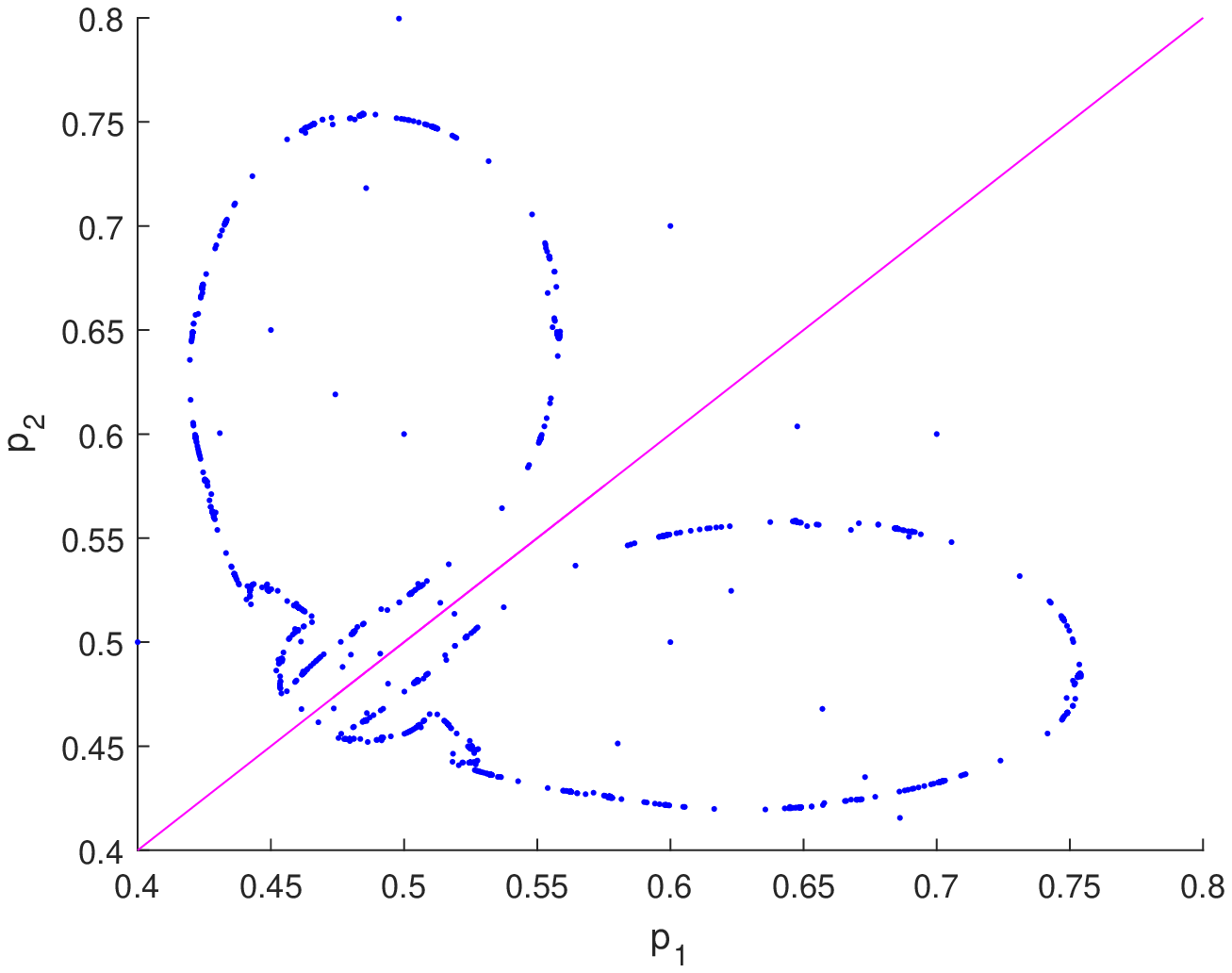}} \\

  \caption{Phase portraits for $k_1=k_2=1$ and $c_1=c_2=0.2$ with the initial point $(0.56,1.06)$. }
 \label{fig:pp-a}
\end{figure}

%\begin{figure}[htbp]
%  \centering
%  {\includegraphics[width=0.7\textwidth]{fig/pp_case_1.eps}} \\
%  \caption{Phase portraits for $k_1=k_2=1,$ $c_1=c_2=0.2,$ $\alpha=0.58.$}
% \label{fig:pp_case_1}
%\end{figure}
%
%\begin{figure}[htbp]
%  \centering
%  {\includegraphics[width=0.7\textwidth]{fig/pp_case_2.eps}} \\
%  \caption{Phase portraits for $k_1=k_2=1,$ $c_1=c_2=0.2,$ $\alpha=0.59.$}
% \label{fig:pp_case_2}
%\end{figure}

%\memo{Fig: 2-D bifurcation diagram w.r.t. $k_1,k_2$: (a) $\alpha=1/2, c_1=0.3, c_2=0.4, (k_1,k_2)\in[0,10]^2$; $\alpha=1/3, c_1=0.1, c_2=0.15, (k_1,k_2)\in[0,6]^2$}

To illustrate the influence of other parameters, several 2-dimensional bifurcation diagrams are computed and displayed in the sequel. Figure \ref{fig:k1k2_0.5_1} depicts the 2-dimensional bifurcation diagram of map \eqref{eq:map-12} ($\alpha=1/2$) with respect to $k_1$ and $k_2$ if we fix $c_1=0.3$, $c_2=0.4$ and set the initial point to be $(0.5,0.8)$. We detect periodic orbits with distinct orders and mark the corresponding parameter points in different colors in Figure \ref{fig:k1k2_0.5_1}. It should be mentioned that the parameter points where there exist periodic orbits with orders more than $25$ are marked in light yellow as well. Two different routes from the unique stable equilibrium to complex dynamics can be observed. For example, if we fix $k_2=7.5$ and change the value of $k_1$ from $0.0$ to $10.0$, the dynamics of the system start from one unique stable equilibrium (the dark blue region), then transition to a stable 2-cycle orbit (the light blue region) and finally to invariant closed orbits as well as chaos (the light yellow region). This is similar to the route displayed in Figure \ref{fig:bifur_a}, where the stable 2-cycle loses its stability through a Neimark-Sacker bifurcation. The other route can be discovered, e.g., if we fix $k_2=2.5$ and keep $k_1$ as a free parameter. Then it is observed that the unique stable equilibrium loses its stability through a cascade of period-doubling bifurcations. 

In Figure \ref{fig:k1k2_0.3_1}, we plot the 2-dimensional bifurcation diagram of map \eqref{eq:map-13} ($\alpha=1/3$) with respect to $k_1$ and $k_2$ if fixing $c_1=0.1$, $c_2=0.15$ and setting the initial point to be $(0.6,0.9)$. Similar to Figure  \ref{fig:k1k2_0.5_1}, the aforementioned two routes from local stability to complex dynamics can also be observed in Figure \ref{fig:k1k2_0.3_1}.

\begin{figure}[htbp]
  \centering
{\includegraphics[width=0.65\textwidth]{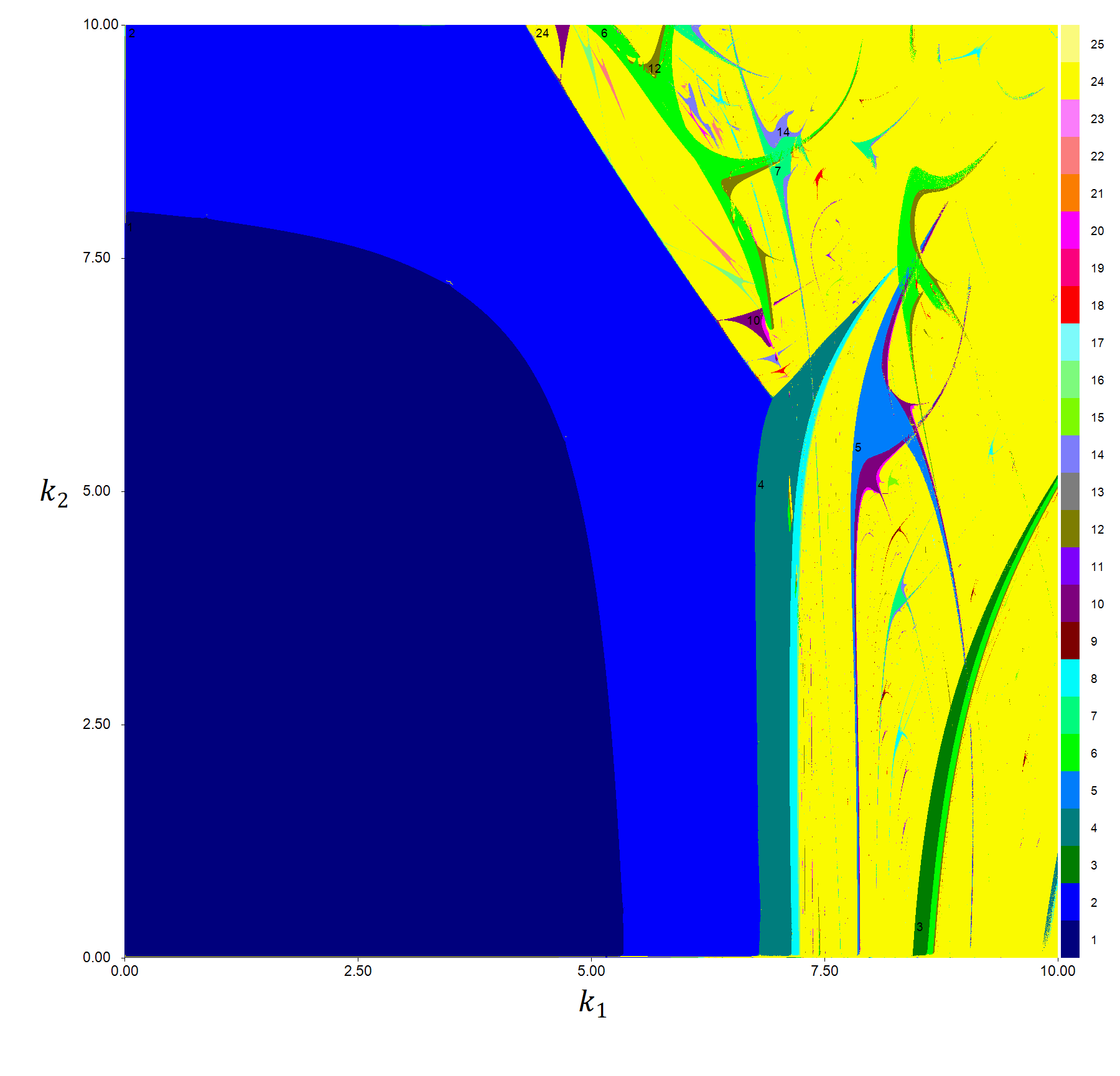}} 
  \caption{The 2-dimensional bifurcation diagram of map \eqref{eq:map-12} ($\alpha=1/2$) with respect to $k_1$ and $k_2$ if we fix $c_1=0.3$, $c_2=0.4$ and set the initial point to be $(0.5,0.8)$.}
 \label{fig:k1k2_0.5_1}
\end{figure}

\begin{figure}[htbp]
  \centering
{\includegraphics[width=0.65\textwidth]{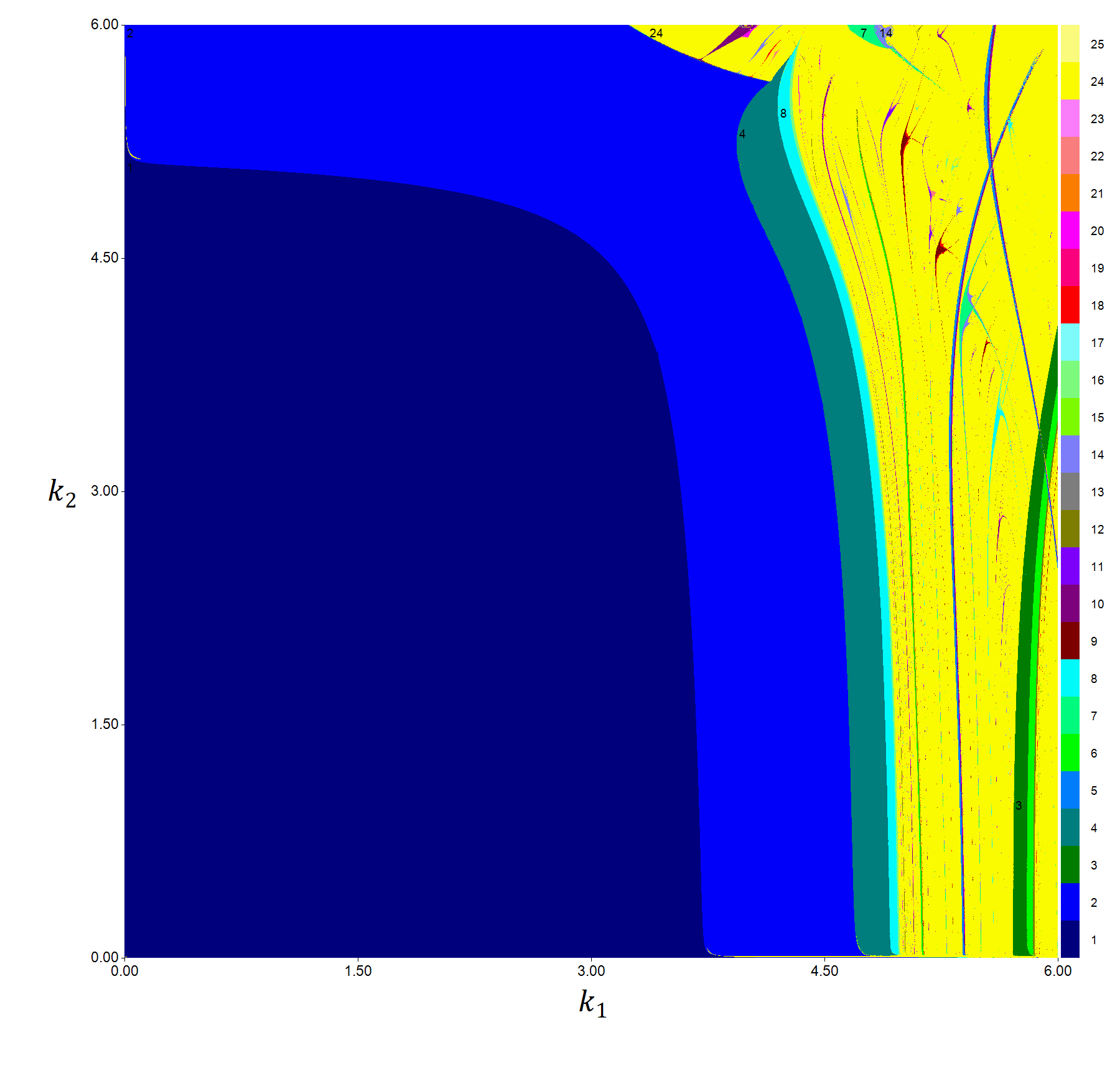}} 
  \caption{The 2-dimensional bifurcation diagram of map \eqref{eq:map-13} ($\alpha=1/3$) with respect to $k_1$ and $k_2$ if we fix $c_1=0.1$, $c_2=0.15$ and set the initial point to be $(0.6,0.9)$.}
 \label{fig:k1k2_0.3_1}
\end{figure}

%\memo{Fig: 2-D bifurcation diagram w.r.t. $c_1,c_2$: (a) $\alpha=1/2, k_1=6, k_2=12, (c_1,c_2)\in[0,1]^2$; $\alpha=1/3, k_1=0.3, k_2=0.6, (c_1,c_2)\in[0,0.1]^2$}

The 2-dimensional bifurcation diagrams with respect to $c_1$ and $c_2$ for $\alpha=1/2$ and $\alpha=1/3$ are displayed in Figures \ref{fig:c1c2_0.5_2} and \ref{fig:c1c2_0.3_2}, respectively. One can see that complicated dynamic phenomena take place if one of the cost parameters $c_1,c_2$ is small enough. Similarly, we find the above two routes to chaotic behavior, i.e., through a cascade of period-doubling bifurcation and through a Neimark-Sacker bifurcation on a 2-cycle orbit, which have already been discovered by Ahmed et al.\ \cite{Ahmed2015O}. However, from Figure \ref{fig:c1c2_0.5_2}, we also find the existence of a Neimark-Sacker bifurcation directly on the unique equilibrium, which is a new result that has not been observed by Ahmed et al.\ \cite{Ahmed2015O} yet. Specifically, Figure \ref{fig:c1c2_0.5_2} shows that, if we fix $c_1=0.9$ and decrease the value of $c_2$ from $1.0$ to $0.0$, the dynamics of the system directly transition from the unique stable equilibrium (the dark blue region) to invariant closed orbits (the light yellow region). In this case, the behavior of the market suddenly changes from an ordered state to a disordered state at some critical point, which can hardly be learned by even rational players.

\begin{figure}[htbp]
  \centering
  {\includegraphics[width=0.65\textwidth]{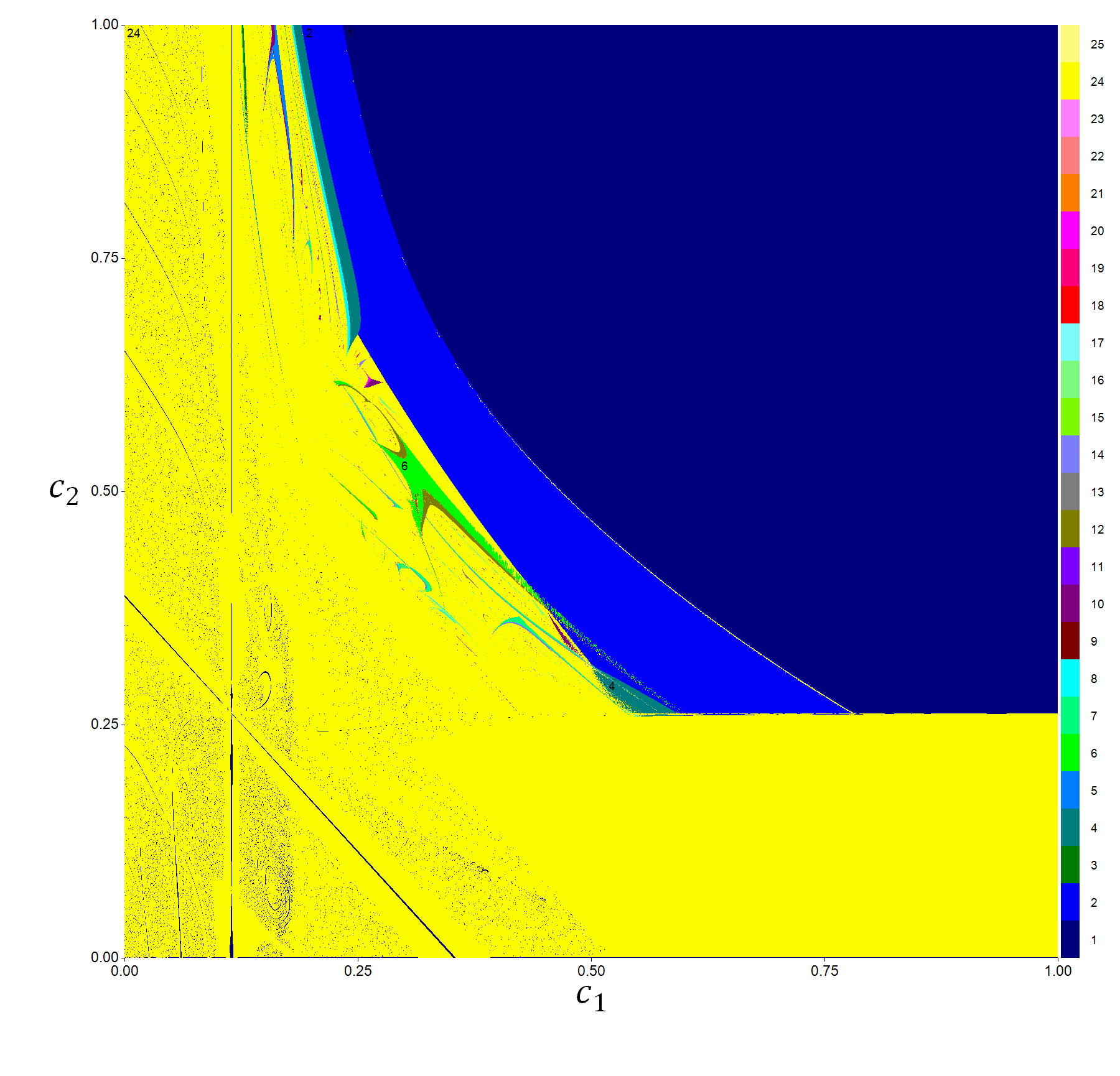}} 
  \caption{The 2-dimensional bifurcation diagram of map \eqref{eq:map-12} ($\alpha=1/2$) with respect to $c_1$ and $c_2$ if we fix $k_1=6$, $k_2=12$ and set the initial point to be $(0.5,0.8)$.}
 \label{fig:c1c2_0.5_2}
\end{figure}

\begin{figure}[htbp]
  \centering
{\includegraphics[width=0.65\textwidth]{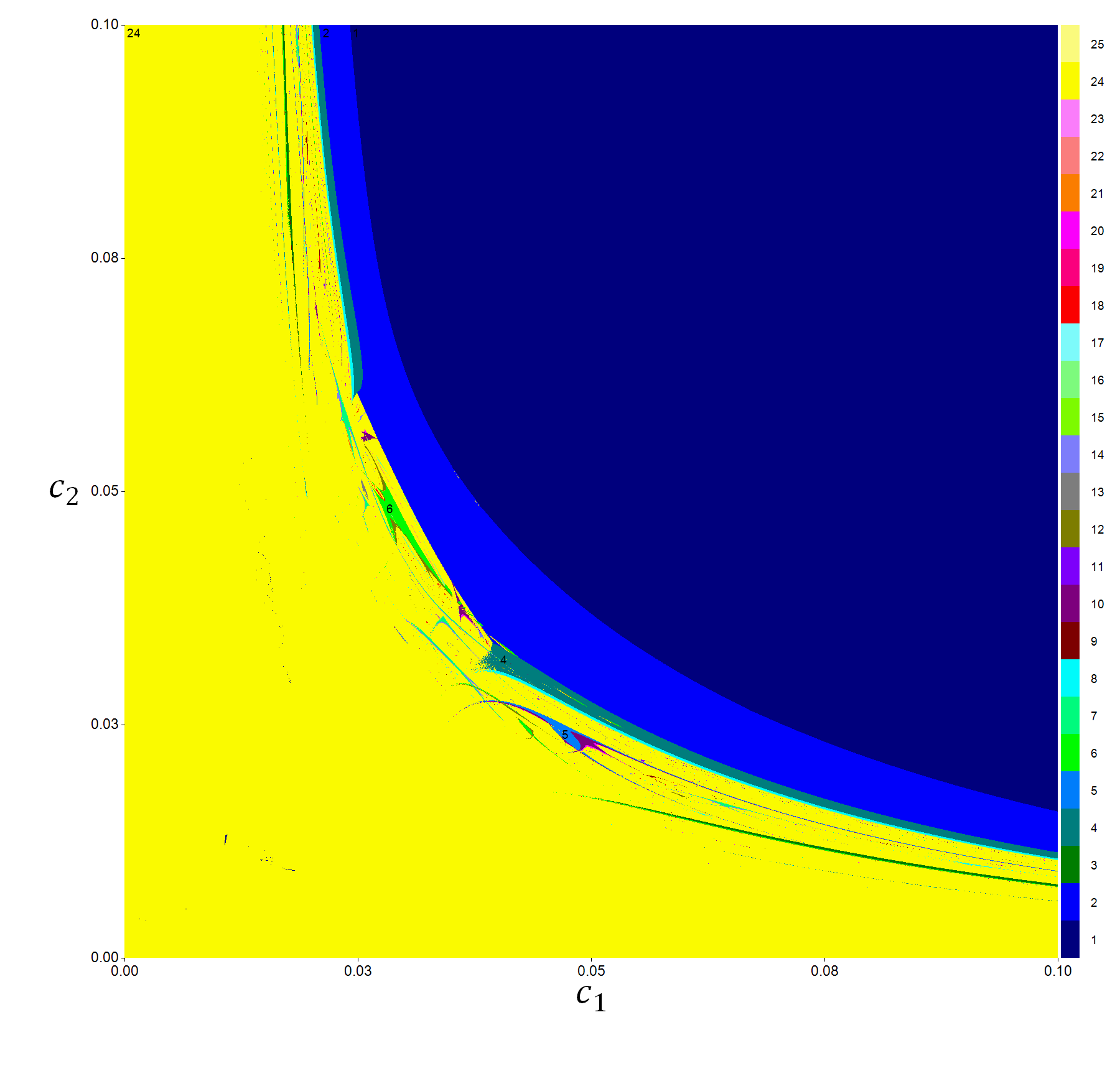}} 
  \caption{The 2-dimensional bifurcation diagram of map \eqref{eq:map-13} ($\alpha=1/3$) with respect to $c_1$ and $c_2$ if we fix $k_1=0.3$, $k_2=0.6$ and set initial point to be $(0.6,0.9)$.}
 \label{fig:c1c2_0.3_2}
\end{figure}

%\gai{\section{Further discussion}}
%\gai{1.所有$x=y$的点收敛于不动点（0.489655172413793，0.489655172413793）；
%2.所有$x\neq y$的点以跳跃的方式最终收敛于包含不动点的两个稳定的不变闭曲线。这里所言的跳跃指的是其轨道依此出现在直线$x=y$两侧，即当前轨道的横坐标大于纵坐标，下一个轨道的横坐标就要小于纵坐标，反之亦然。}

\section{Concluding Remarks}

In this paper, we investigated the local stability, bifurcations, and comparative statics of a dynamic Bertrand duopoly game with differentiated products. This duopoly is assumed to possess two boundedly rational players adopting a gradient adjustment mechanism and a continuum of identical consumers with a CES utility function. Moreover, the cost functions are supposed to be linear. It should be mentioned that the nonlinearity of the resulting demand function derived from the underlying utility permits us to extend the applications of Bertrand games to more realistic economies, compared to the widely used Bertrand models with linear demands.

The considered game was first explored by Ahmed et al.\ \cite{Ahmed2015O}, where only numerical simulations are employed to investigate the dynamic behavior and it was observed that the Nash equilibrium loses its stability through a period-doubling bifurcation as the speed of adjustment increases. In our study, however, we re-investigated this game using several tools based on symbolic computations such as the triangular decomposition method (refer to, e.g., \cite{Li2010D}) and the PCAD method (refer to, e.g., \cite{Collins1991P}). The results of symbolic computations are exact, and thus provide theoretical foundations for the systematic analysis of economic models.

For simplicity, our work mainly focused on two specific degrees of product substitutability, namely $\alpha=1/2$ and $\alpha=1/3$. In both cases, we proved the uniqueness of the non-vanishing equilibrium using the algebraic approach of detecting the multiplicity of equilibria proposed by the first author and his co-worker \cite{Li2014C}. We introduce several tools based on symbolic computations and used them to obtain the rigorous conditions for the local stability of the unique non-vanishing equilibrium for the first time. In the special case that the two firms have identical marginal costs, we proved that the model can lose its stability only through a period-doubling bifurcation. From an economic point of view, the most interesting finding was that an increase in the substitutability degree or a decrease in the product differentiation leads to the destabilization of the Bertrand model. This is because a price war, which might destabilize the equilibrium, can take place if the substitutability degree is large enough. We should mention that our finding is in contrast with that by Agliari et al.\ \cite{Agliari2016N} and that by Fanti and Gori \cite{Fanti2012T}. This contradiction contributes to the literature on the connection between Cournot and Bertrand oligopolies and may help reveal the essential difference between them. Moreover, we conducted the comparative statics in the special case of identical marginal costs. The resulting conclusion was that lower degrees of product differentiation mean lower prices, higher supplies, lower profits, and lower social welfare, which is consistent with our economic intuition. 

Numerical simulations were provided in the end, through which complex dynamics such as periodic orbits and chaos can be observed. The simulations confirmed that an increase in the substitutability degree $\alpha$ leads to the emergence of instability, complex dynamics, and even chaos in the considered model. Two-dimensional bifurcation diagrams were also provided to show different possible routes to chaotic behavior, e.g., through a cascade of period-doubling bifurcation and through a Neimark-Sacker bifurcation on a 2-cycle orbit. Furthermore, we discovered the existence of a Neimark-Sacker bifurcation directly on the equilibrium, which is a new finding and has not yet been discovered by Ahmed et al.\ \cite{Ahmed2015O}.

%\section*{Appendix A}
%
%\subsection*{A.1. Proof of Theorem \ref{thm:stability12-c}}
%
%\subsection*{A.2. Proof of Proposition \ref{prop:bifur-12-k}}
%
%\subsection*{A.3. Proof of Theorem \ref{thm:stability13-c}}
%
%\subsection*{A.4. Proof of Proposition \ref{prop:bifur-13-k}}
%
%\subsection*{A.5. Proof of Proposition \ref{prop:reg-compare-23}}

\section*{Appendix}

\begin{align*}
\begin{autobreak}
R_{1} = 
15552\,c_1^{10} c_2^{6}
+62208\,c_1^{9} c_2^{7}
+93312\,c_1^{8} c_2^{8}
+62208\,c_1^{7} c_2^{9}
+15552\,c_1^{6} c_2^{10}
+73728\,c_1^{11} c_2^{3} k
+327168\,c_1^{10} c_2^{4} k
+576576\,c_1^{9} c_2^{5} k
+541440\,c_1^{8} c_2^{6} k
+436608\,c_1^{7} c_2^{7} k
+541440\,c_1^{6} c_2^{8} k
+576576\,c_1^{5} c_2^{9} k
+327168\,c_1^{4} c_2^{10} k
+73728\,c_1^{3} c_2^{11} k
+32768\,c_1^{11} c_2\,k^{2}
+94208\,c_1^{10} c_2^{2} k^{2}
+284160\,c_1^{9} c_2^{3} k^{2}
+1163712\,c_1^{8} c_2^{4} k^{2}
+2855520\,c_1^{7} c_2^{5} k^{2}
+3825168\,c_1^{6} c_2^{6} k^{2}
+2855520\,c_1^{5} c_2^{7} k^{2}
+1163712\,c_1^{4} c_2^{8} k^{2}
+284160\,c_1^{3} c_2^{9} k^{2}
+94208\,c_1^{2} c_2^{10} k^{2}
+32768\,c_1c_2^{11} k^{2}
+77824\,c_1^{9} c_2\,k^{3}
+ 359936\,c_1^{8} c_2^{2} k^{3}
+644608\,c_1^{7} c_2^{3} k^{3}
+610976\,c_1^{6} c_2^{4} k^{3}
+494368\,c_1^{5} c_2^{5} k^{3}
+610976\,c_1^{4} c_2^{6} k^{3}
+644608\,c_1^{3} c_2^{7} k^{3}
+359936\,c_1^{2} c_2^{8} k^{3}
+77824\,c_1c_2^{9} k^{3}
-4096  c_1^{8} k^{4}
-12288\,c_1^{7} c_2\,k^{4}
+4544\,c_1^{6} c_2^{2} k^{4}
+70360\,c_1^{5} c_2^{3} k^{4}
+114600\,c_1^{4} c_2^{4} k^{4}
+70360\,c_1^{3} c_2^{5} k^{4}
+4544\,c_1^{2} c_2^{6} k^{4}
-12288\,c_1c_2^{7} k^{4}
-4096\,c_2^{8} k^{4}
-1024\,c_1^{5} c_2\,k^{5}
-3232\,c_1^{4} c_2^{2} k^{5}
-4488\,c_1^{3} c_2^{3} k^{5}
-3232\,c_1^{2} c_2^{4} k^{5}
-1024\,c_1c_2^{5} k^{5}
+32\,c_1^{3} c_2\,k^{6}
+61\,c_1^{2} c_2^{2} k^{6}
+32\,c_1c_2^{3} k^{6},
\end{autobreak}\\

\begin{autobreak}
R_{2} = 
1152\,c_1^{8} c_2^{2}
+5832\,c_1^{7} c_2^{3}
+12960\,c_1^{6} c_2^{4}
+16560\,c_1^{5} c_2^{5}
+12960\,c_1^{4} c_2^{6}
+5832\,c_1^{3} c_2^{7}
+1152\,c_1^{2} c_2^{8}
+1024\,c_1^{8} k
+3584\,c_1^{7} c_2\,k
+5920\,c_1^{6} c_2^{2} k
+6224\,c_1^{5} c_2^{3} k 
+5836\,c_1^{4} c_2^{4} k
+6224\,c_1^{3} c_2^{5} k
+5920\,c_1^{2} c_2^{6} k
+3584\,c_1c_2^{7} k
+1024\,c_2^{8} k
+512\,c_1^{5} c_2\,k^{2}
+1616\,c_1^{4} c_2^{2} k^{2}
+2244\,c_1^{3} c_2^{3} k^{2}
+1616\,c_1^{2} c_2^{4} k^{2}
+512\,c_1c_2^{5} k^{2}
- 32\,c_1^{3} c_2\,k^{3}
-61\,c_1^{2} c_2^{2} k^{3}
-32\,c_1c_2^{3} k^{3}, 
\end{autobreak}\\

\begin{autobreak}
R_3 = 
-209715200000\,c_1^{12} c_2^{8}
+838860800000\,c_1^{11} c_2^{9}
-1258291200000\,c_1^{10} c_2^{10}
+838860800000\,c_1^{9} c_2^{11}
-209715200000\,c_1^{8} c_2^{12}
+1160950579200\,c_1^{13} c_2^{5} k
-5170397184000\,c_1^{12}  c_2^{6} k
+9284105011200\,c_1^{11} c_2^{7} k
-9178054656000\,c_1^{10} c_2^{8} k
+7806792499200\,c_1^{9} c_2^{9} k
-9178054656000\,c_1^{8} c_2^{10} k
+9284105011200\,c_1^{7} c_2^{11} k
-5170397184000\,c_1^{6} c_2^{12} k
+ 1160950579200\,c_1^{5} c_2^{13} k
+626913312768\,c_1^{13} c_2^{3} k^{2}
-1827529703424\,c_1^{12} c_2^{4} k^{2}
+6377496477696\,c_1^{11} c_2^{5} k^{2}
-24562717922304\,c_1^{10} c_2^{6} k^{2}
+56911413825536\,c_1^{9} c_2^{7} k^{2}
- 74841436780544\,c_1^{8} c_2^{8} k^{2}
+56911413825536\,c_1^{7} c_2^{9} k^{2}
-24562717922304\,c_1^{6} c_2^{10} k^{2}
+6377496477696\,c_1^{5} c_2^{11} k^{2}
-1827529703424\,c_1^{4} c_2^{12} k^{2}
+626913312768\,c_1^{3} c_2^{13} k^{2} 
-117546246144\,c_1^{12} c_2^{2} k^{3}
+2268751389696\,c_1^{11} c_2^{3} k^{3}
-8446241806848\,c_1^{10} c_2^{4} k^{3}
+13848228389376\,c_1^{9} c_2^{5} k^{3}
-12871123435008\,c_1^{8} c_2^{6} k^{3}
+10570707526656\,c_1^{7} c_2^{7} k^{3} 
-12871123435008\,c_1^{6} c_2^{8} k^{3}
+13848228389376\,c_1^{5} c_2^{9} k^{3}
-8446241806848\,c_1^{4} c_2^{10} k^{3}
+2268751389696\,c_1^{3} c_2^{11} k^{3}
-117546246144\,c_1^{2} c_2^{12} k^{3}
+7346640384\,c_1^{11} c_2\,k^{4}
+ 23872802112\,c_1^{10} c_2^{2} k^{4}
-79144786368\,c_1^{9} c_2^{3} k^{4}
-389232360000\,c_1^{8} c_2^{4} k^{4}
+1762366805056\,c_1^{7} c_2^{5} k^{4}
-2639431381760\,c_1^{6} c_2^{6} k^{4}
+1762366805056\,c_1^{5} c_2^{7} k^{4}
- 389232360000\,c_1^{4} c_2^{8} k^{4}
-79144786368\,c_1^{3} c_2^{9} k^{4}
+23872802112\,c_1^{2} c_2^{10} k^{4}
+7346640384\,c_1c_2^{11} k^{4}
-153055008\,c_1^{10} k^{5}
+444048480\,c_1^{9} c_2\,k^{5}
-2281361760\,c_1^{8} c_2^{2} k^{5}
- 6359031360\,c_1^{7} c_2^{3} k^{5}
+33853070112\,c_1^{6} c_2^{4} k^{5}
-51945109632\,c_1^{5} c_2^{5} k^{5}
+33853070112\,c_1^{4} c_2^{6} k^{5}
-6359031360\,c_1^{3} c_2^{7} k^{5}
-2281361760\,c_1^{2} c_2^{8} k^{5}
+444048480\,c_1c_2^{9} k ^{5}
-153055008\,c_2^{10} k^{5}
+36636624\,c_1^{7} c_2\,k^{6}
-65578896\,c_1^{6} c_2^{2} k^{6}
+239834412\,c_1^{5} c_2^{3} k^{6}
-377249916\,c_1^{4} c_2^{4} k^{6}
+239834412\,c_1^{3} c_2^{5} k^{6}
-65578896\,c_1^{2} c_2^{6} k^{6}
+36636624  c_1c_2^{7} k^{6}
-669222\,c_1^{5} c_2\,k^{7}
+1023534\,c_1^{4} c_2^{2} k^{7}
-951468\,c_1^{3} c_2^{3} k^{7}
+1023534\,c_1^{2} c_2^{4} k^{7}
-669222\,c_1c_2^{5} k^{7}
+2187\,c_1^{3} c_2\,k^{8}
-4031\,c_1^{2} c_2^{2} k^{8}
+2187\,c_1c_2^{3} k^{8},
\end{autobreak}\\

\begin{autobreak}
R_4 = 
17714700\,c_1^{10} c_2^{2}
-84798900\,c_1^{9} c_2^{3}
+166819500\,c_1^{8} c_2^{4}
-187523100\,c_1^{7} c_2^{5}
+175575600\,c_1^{6} c_2^{6}
-187523100\,c_1^{5} c_2^{7}
+166819500\,c_1^{4} c_2^{8}
-84798900\,c_1^{3} c_2^{9}
+17714700\,c_1^{2}  c_2^{10}
+19131876\,c_1^{10} k
-55506060\,c_1^{9} c_2\,k
+70441812\,c_1^{8} c_2^{2} k
-70683840\,c_1^{7} c_2^{3} k
+106503012\,c_1^{6} c_2^{4} k
-136123200\,c_1^{5} c_2^{5} k
+106503012\,c_1^{4} c_2^{6} k
-70683840\,c_1^{3} c_2^{7} k
+ 70441812\,c_1^{2} c_2^{8} k
-55506060\,c_1c_2^{9} k
+19131876\,c_2^{10} k
-9159156\,c_1^{7} c_2\,k^{2}
+23480604\,c_1^{6} c_2^{2} k^{2}
-24625107\,c_1^{5} c_2^{3} k^{2}
+19286271\,c_1^{4} c_2^{4} k^{2}
-24625107\,c_1^{3} c_2^{5} k^{2}
+ 23480604\,c_1^{2} c_2^{6} k^{2}
-9159156\,c_1c_2^{7} k^{2}
+334611\,c_1^{5} c_2\,k^{3}
-511767\,c_1^{4} c_2^{2} k^{3}
+475734\,c_1^{3} c_2^{3} k^{3}
-511767\,c_1^{2} c_2^{4} k^{3}
+334611\,c_1c_2^{5} k^{3}
-2187\,c_1^{3} c_2\,k^{4}
+4031\,c_1^{2}  c_2^{2} k^{4}
-2187\,c_1c_2^{3} k^{4},
\end{autobreak}\\

\begin{autobreak}
A_{1} = 
243\,c_1^{2}
+352\,c_1c_2
-9\,k, 
\end{autobreak}\\

\begin{autobreak}
A_{2} = 
-8000\,c_1^{5} c_2^{3}
+19683\,c_1^{6} k
-17496\,c_1^{5} c_2\,k
+3024\,c_1^{4} c_2^{2} k
+1728\,c_1^{3} c_2^{3} k
-2187\,c_1^{4} k^{2}
+3564\,c_1^{3} c_2\,k^{2}
-432\,c_1^{2} c_2^{2} k^{2}
+81\,c_1^{2} k^{3}
+36\,c_1c_2\,k^{3}
-k^{4},
\end{autobreak}\\

\begin{autobreak}
A_{3} = 
12754584\,c_1^{7}
-12171384\,c_1^{6} c_2
+3708504\,c_1^{5} c_2^{2}
+84096\,c_1^{4} c_2^{3}
+2519424\,c_1^{3} c_2^{4}
-72171\,c_1^{5} k
-3576744\,c_1^{4} c_2\,k
+5126856\,c_1^{3} c_2^{2} k
-629856\,c_1^{2} c_2^{3} k
-25272\,c_1^{3} k^{2} 
+98966\,c_1^{2} c_2\,k^{2}
+52488\,c_1c_2^{2} k^{2}
+387\,c_1\,k^{3}
-1458\,c_2\,k^{3}.
\end{autobreak}

\end{align*}

\section*{Acknowledgements}
The authors wish to thank Dr.\ Li Su for the beneficial discussions.
The authors are grateful to the anonymous referees for their helpful comments.

This work has been supported by Philosophy and Social Science Foundation of Guangdong (Grant No.\ GD21CLJ01), Natural Science Foundation of Anhui Province (Grant No.\ 2008085QA09), University Natural Science Research Project of Anhui Province (Grant No.\ KJ2021A0482), Major Research and Cultivation Project of Dongguan City University (Grant No.\ 2021YZDYB04Z).

\bibliographystyle{abbrv}
\bibliography{ref.bib}

\end{CJK}
\end{document}

%% file: macro.tex
\usepackage{CJKutf8}
\usepackage{authblk}
\usepackage{listings}

\usepackage{amssymb}
\usepackage{amsmath}
\usepackage{amsthm}
\usepackage{amsfonts}
\usepackage{mathtools}
\usepackage{mathrsfs}
\usepackage{bm} % 处理数学公式中的黑斜体的宏包

\usepackage{graphicx,subfig}
\usepackage{dcolumn}
\usepackage[ruled,lined,algonl]{algorithm2e}
\usepackage{color}
\usepackage{autobreak}
\allowdisplaybreaks

\theoremstyle{plain}
\newtheorem{theorem}{Theorem}

\newtheorem{lemma}{Lemma}
\newtheorem{proposition}{Proposition}

\theoremstyle{definition}
\newtheorem{definition}{Definition}

\theoremstyle{remark}
\newtheorem{remark}{Remark}

\newcommand{\pset}[1]{\mathcal{#1}}
\DeclareMathOperator{\numer}{Num}
\DeclareMathOperator{\denom}{Den}
\DeclareMathOperator{\Det}{Det}
\DeclareMathOperator{\Tr}{Tr}
\DeclareMathOperator{\res}{res}